\documentclass[10pt,journal,compsoc]{IEEEtran}

\usepackage{graphicx}
\usepackage{caption}
\usepackage{subcaption}
\usepackage{array}
\usepackage{multirow}
\usepackage{amsthm}
\usepackage{amsmath}
\usepackage{amsfonts}
\usepackage{amssymb}
\usepackage{graphicx}
\usepackage{multirow}
\usepackage{setspace}
\usepackage{balance}
\usepackage{graphicx}
\usepackage{caption}
\usepackage{subcaption}
\usepackage{algorithmic}
\usepackage{algorithm}
\usepackage{comment}
\usepackage{color}
\usepackage{xr}
\usepackage{pdfpages}

\usepackage[labelformat=simple]{subcaption}

\newtheorem{thm}{Theorem}
\newtheorem{lem}{Lemma}

\newtheorem*{remark}{Remark}
\newtheorem{coro}[thm]{Corollary}

\ifCLASSOPTIONcompsoc
  \usepackage[nocompress]{cite}
\else
  \usepackage{cite}
\fi
\ifCLASSINFOpdf
\else
\fi
\begin{document}
%
\title{\huge{Design and Analysis of Dynamic Auto Scaling Algorithm (DASA) for virtual EPC (vEPC) in 5G Networks}}
%
%
%
%

\author{Yi Ren, \textit{Member, IEEE},  Tuan Phung-Duc,   and Jyh-Cheng Chen, \textit{Fellow, IEEE}
\IEEEcompsocitemizethanks{\IEEEcompsocthanksitem A preliminary version of this paper was presented at 59th IEEE GLOBECOM, 2016, in Washington DC, USA.
\IEEEcompsocthanksitem Y. Ren, and J.-C. Chen are with Department of Computer Science, National Chiao Tung University, Taiwan, e-mail: \{renyi, jcc\}@cs.nctu.edu.tw
\IEEEcompsocthanksitem T. Phung-Duc is with Faculty of Engineering, Information and Systems, University of Tsukuba, Ibaraki, Japan, e-mail:tuan@sk.tsukuba.ac.jp.}
}

\IEEEtitleabstractindextext{%
\begin{abstract}
Network Function Virtualization (NFV) enables mobile operators to virtualize their network entities as Virtualized Network Functions (VNFs), offering fine-grained on-demand network capabilities. VNFs can be dynamically scale-in/out to meet the performance requirements for future 5G networks. However, designing an auto-scaling algorithm with low operation cost and low latency while considering the capacity of legacy network equipment is a challenge. In this paper, we propose a VNF Dynamic Auto Scaling Algorithm (DASA) considering the tradeoff between performance and operation cost. We also develop an analytical model to quantify the tradeoff and validate the analysis through extensive simulations. The system is modeled as a queueing model while legacy network equipment is considered as a reserved block of servers. The VNF instances are powered on and off according to the number of job requests. The results show that the proposed DASA can significantly reduce operation cost given the latency upper-bound. Moreover, the models provide a quick way to evaluate the cost-performance tradeoff without wide deployment, which can save cost and time.
\end{abstract}

}

\maketitle

\IEEEdisplaynontitleabstractindextext

%
\IEEEpeerreviewmaketitle

\IEEEraisesectionheading{\section{Introduction}\label{sec:Introduction}}

\IEEEPARstart{C}{ellular}  networks have been evolved to 4th generation (4G). Long Term Evolution-Advanced (LTE-A) has become a commonly used communication technology worldwide and is continuously expanding and evolving to 5th generation (5G). A recent report states that 95 operators have commercially launched LTE-A networks in 48 countries and the total smartphone traffic is expected to rise 11 times from 2015 to 2021~\cite{Ericsson2015}. Accordingly, operators are improving their network infrastructure to increase capacity and to meet the demand for fast-growing data traffic in 5G networks.

One of the most important technologies for 5G networks is to utilize Network Function Virtualization (NFV) to virtualize the network components in the core network which is called Evolved Packet Core (EPC). The virtualized EPC is commonly referred to as virtual EPC (vEPC)~\cite{ETSIGSNFVINF2015}. The emergence of NFV enables operators to manage their network equipment in a fine-grained and efficient way~\cite{hawilo2014nfv}. Indeed, legacy network infrastructure suffers from the nature that data traffic usually has peaks during a day while having relative low utilization in the rest of time (e.g., in the midnight).  To guarantee the Quality of user Experience (QoE), operators usually leave spare capacities to tackle the peak traffic while deploying network equipment. Accordingly, the network equipment are under low utilization during non-busy periods. NFV enables operators to virtualize hardware resources. It also makes special-purpose network equipment toward software solutions, i.e., Virtualized Network Function (VNF) instances. A VNF instance can run on several Virtual Machines (VMs) which can be \textit{scaled-out/in} to adjust the VNF's computing and networking capabilities to save both energy and resources. Although the idea is just being applied to cellular networks, it has been used in the community of cloud computing. A classic case is Animoto, an image-processing service provider, experienced a demand of surging from 50 VM instances to 4000 VM instances (Amazon EC2 instances) in three days in April 2008. After the peak, the demand fell sharply to an average level~\cite{Animoto}. Animoto only paid for 4000 instances for the peak time. In future 5G networks, it is expected that there will be heterogeneous types of traffic, including traffic from Human-to-Human (H2H) and Machine-to-Machine (M2M) communications. With such diverse traffic types, it is very likely similar case as that in Animoto will also happen to future 5G networks.

Given the fact that \textit{auto-scaling} VNF instance can decrease operation cost while meeting the demand for VNF service, it is critical to design good strategies to allocate VNF instances adaptively to fulfill the demands of service requirements. However, it is not a trivial task. Specifically, the operation cost is reduced by decreasing the number of powered-on VNF instances. On the other hand, resource under-provisioning may cause Service Level Agreement (SLA) violations. Therefore, the goal of a desirable strategy is to reduce operation cost while also maintaining acceptable levels of performance. Thus, a \textit{cost-performance tradeoff} is formed: The VNF performance is improved by \textit{scaling-out} VNF instances while the operation cost is reduced by \textit{scaling-in} VNF instances.

Given that legacy equipment usually is expensive in cellular networks, network operators usually power on network equipment all the time and try to use the equipment as long as they can to maximize Return On Investment (ROI). Thus, nowadays most operators operate old generation systems and new generation systems simultaneously. For example, many operators offer 3G and 4G services at the same time. In future 5G systems, virtualized resources would be added to boost the system performance of legacy 4G systems, leading to the evolution from EPC to vEPC\footnote{Details will be discussed in Section~\ref{sec:Background}.}. It is expected that 5G and legacy systems will coexist.

In this paper, we study the cost-performance tradeoff while considering both the \textit{VM setup time} and \textit{legacy equipment capacity}. In~\cite{xiao2013dynamic,jokhio2013prediction,roy2011efficient,tirado2011predictive, niu2012quality,huang2013migration,huang2014prediction,calheiros2014workload, islam2012empirical,bashar2013autonomic,bankole2013cloud,  shen2011cloudscale,khan2012workload}, they either ignore VM setup time or only consider virtualized resource itself while overlooking legacy (fixed) resources. However, this is not practical for future 5G cellular networks as follows:

\begin{itemize}
	
	\item Although a scale-out request can be sent right way, a VNF instance cannot be available immediately.  The lag time could be as long as 10 minutes or more to start an instance in Microsoft Azure and it could be varied from time to time~\cite{hill2010early}. It could happen that the instance is too late to serve the VNF if the lag time is not taken into consideration.
	
	\item The capacity of legacy network equipment is another issue. For example, a network operator with legacy network equipment wants to increase network capacities by using NFV technique. The desired solution should consider the capacities of both legacy network equipment and VNFs. If the capacity of a legacy network equipment is only equal to that of one VNF,	scaling-out from one VNF instance to two VNF instances increases $100\%$ capacity. However, if the capacity of a legacy network equipment is equal to that of 100 VNF, its capacity only grows less than $1\%$ when adding one more VNF instance. Current cloud auto-scaling schemes usually ignore this problem which is called \textit{non-constant} issue~\cite{mao2010cloud}. In other words, the capacity of legacy network equipment has a significant impact on the desired auto-scaling solution for future 5G systems.
	
\end{itemize}

In this paper, we propose Dynamic Auto Scaling Algorithm (DASA) to solve the problems.  In the proposed DASA, we consider that legacy 4G network equipment is powered on all the time as a block, while virtualized resources (VNF instances) are added to or deleted from the system dynamically. To the best of our knowledge, this has not been discussed in any previous literature. The VNF instances are scaled in and out depending on the number of jobs in the system. A critical issue is how to specify a suitable $k$, the number of VNF instances, for the cost-performance tradeoff. We propose detailed analytical models to answer this question. The cost-performance tradeoff is quantified as \textit{operation cost metric} and  \textit{performance metric} of which closed-form solutions are derived and validated against extensive discrete-event simulations. Moreover, we develop a recursive algorithm that reduces the computational complexity from  $O(k^3 \times K^3)$ to $O(k \times K)$, where $K$ is the total capacity of the system. Without our algorithm, it is difficult to solve the problem in a short time. Our models enable wide applicability in various scenarios, and therefore, have important theoretical significance. Furthermore, this work offers network operators guidelines to design an optimal VNF auto-scaling strategy based on their management policies in a systematical way.

The rest of this paper is organized as follows. In Section~\ref{sec:Background}, we  briefly introduce some background on mobile networks and NFV architecture. Section~\ref{sec:Related_Work} reviews the related work. Challenges and contributions are addressed in Section~\ref{sec:chall}. In Section~\ref{sec:Proposed_Algorithm}, we presents the proposed DASA for VNF auto-scaling applications, followed by numerical results illustrated in Section~\ref{sec:Simulation_Results}.   Section~\ref{sec:Conclusions} concludes this paper.

\section{Background} \label{sec:Background}
\begin{figure}
	\centering
	\includegraphics[width=9cm]{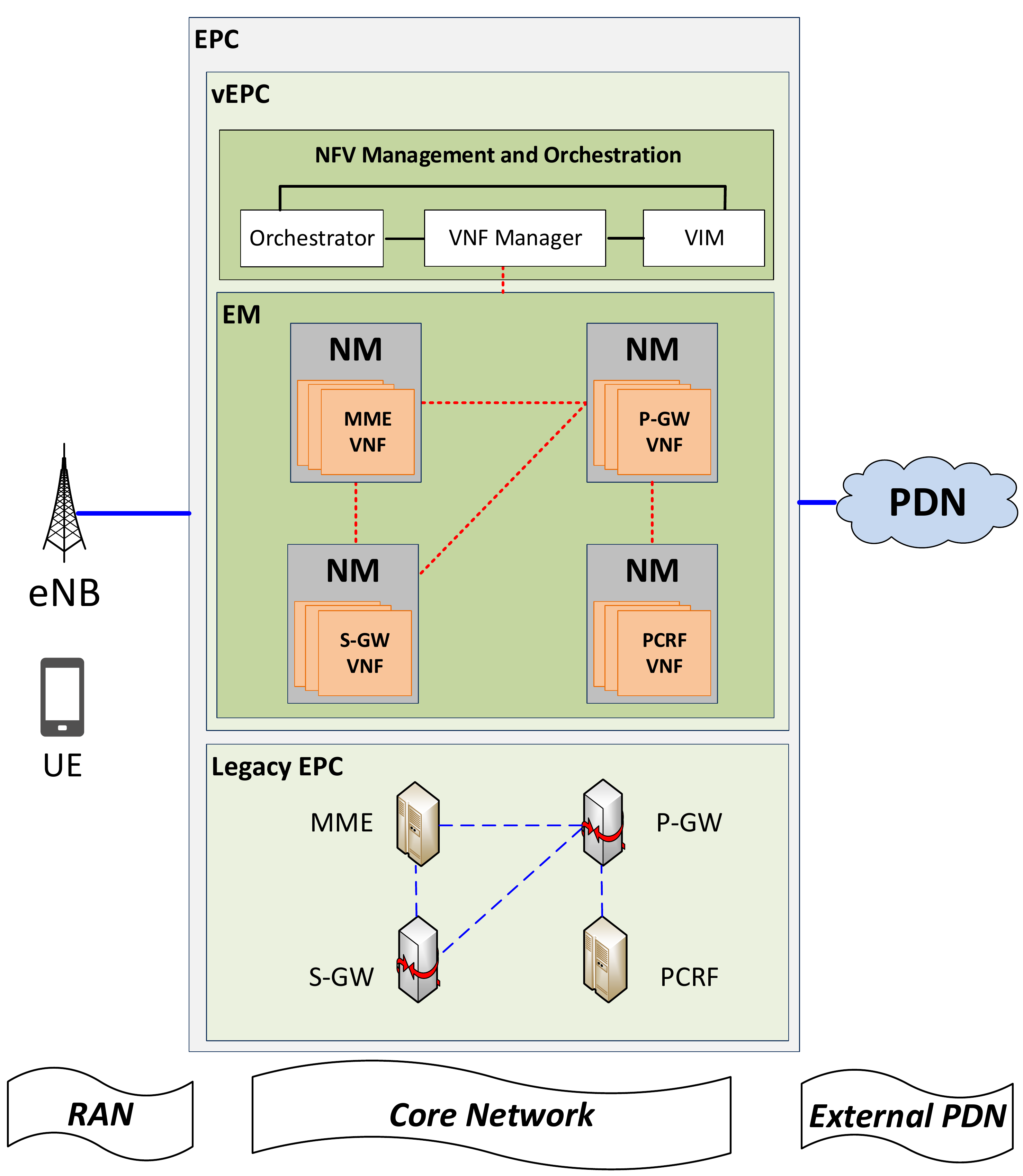}
	\caption{A simplified example of NFV enabled LTE architecture.}
	\label{fig:vEPC_arch}
	\vspace{-5mm}
\end{figure}
A cellular network typical is composed of Radio Access Network (RAN) and  Core Network (CN) as shown in Fig.~\ref{fig:vEPC_arch}~\cite{chen2004ip}. Starting from release 8 in 3GPP, the RAN and CN are referred to as Evolved-UTRAN (E-UTRAN) and EPC, respectively. The main target of NFV we consider in this paper is to virtualize  the functions in the EPC. Here, we use an example to explain EPC and vEPC when NFV is deployed. Fig.~\ref{fig:vEPC_arch}  shows a simplified example of NFV enabled LTE architecture which consists of RAN, EPC, and external Packet Data Network (PDN). In particular,  the EPC is composed of legacy EPC and vEPC. In the following, we brief introduce them.

\subsection{Legacy EPC}

In E-UTRAN, a User Equipment (UE) connects to EPC through an eNB, which essentially is a base station. Here, we only show basic functions in EPC, including Serving Gateway (S-GW), PDN Gateway (P-GW), Mobility Management Entity (MME), and Policy and Charging Rules Function (PCRF) in the EPC. Please refer to~\cite{TS23.002.2016} for details.

The P-GW is a gateway providing connectivity between EPC and an external PDN. The S-GW is responsible for user data functions enabling routing and packet forwarding to the P-GW. MME handles UE mobility and other control functions. PCRF is a policy and charging control element for policy enforcement, flow-based charging, and service data flow detection.

\subsection{vEPC}
Generally, the vEPC can be divided into two main components: \textit{NFV Management and Orchestration} and \textit{Element Manager (EM)} which are parts of 3GPP Management Reference Model~\cite{V13.1.02015}.

\subsubsection{NFV Management and Orchestration}
The \textit{NFV Management and Orchestration} consists of orchestrator, VNF manager, and Virtualized Infrastructure Manager (VIM). It controls the lifecycles of VNFs and decides whether a VNF should be scaled-out/in/up/down. Additionally, it manages both hardware and software resources to support VNFs. In other words, it can be considered as a bridge between network resources and VNFs. The actions of VNF scaling are described as follows.

\begin{itemize}
	\item VNF scale-in/out: As shown in Fig.~\ref{fig:VNF_example}, a VNF can have many VNF instances, inside which there may be many VMs. VNF scale-out refers to increase the number of VNF instances. In contrast, VNF scale-in is an action to remove existing VNF instances in a sense that virtualized hardware resources are freed and no longer needed.
	\item VNF scale-up/down: VNF scale-up allocates more VMs into an existing VNF instance. Whereas, VNF scale-down releases some VMs from an existing VNF instance.
\end{itemize}

\begin{figure}[t]
	\centering
	\includegraphics[width=5.5cm]{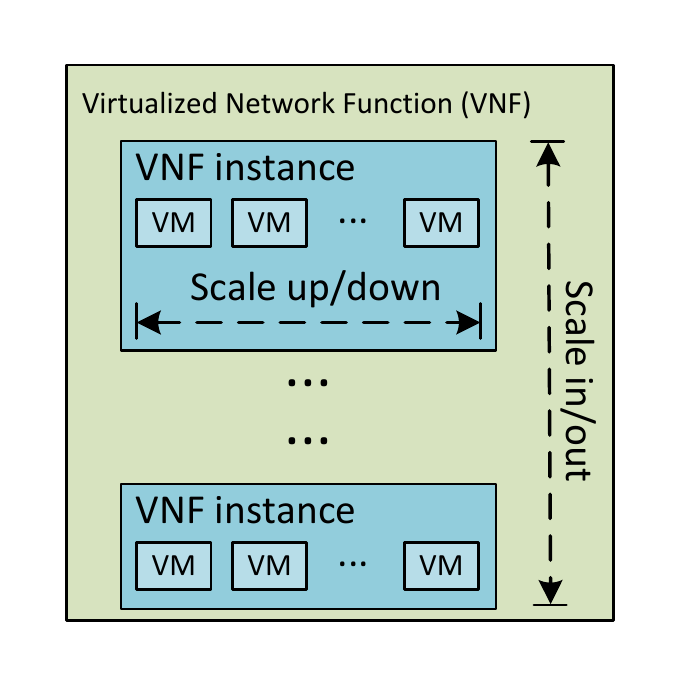}
	\caption{VNF scale up/down and scale in/out.}
	\label{fig:VNF_example}
\end{figure}

\subsubsection{Element Manager (EM)}
3GPP introduces NFV management functions and solutions for mobile core networks based on ETSI NFV specification~\cite{V13.1.02015}. In vEPC, each Network Element (NE) in legacy EPC such as S-GW, P-GW, MME, and PCRF is virtualized as a VNF.  As shown in Fig.~\ref{fig:vEPC_arch}, a Network Manager (NM) provides end-user functions for network management for each NE. Element Manager (EM) is responsible for the management of a set of NMs.

\begin{figure}[t]
	\centering
	\includegraphics[width=8.7cm]{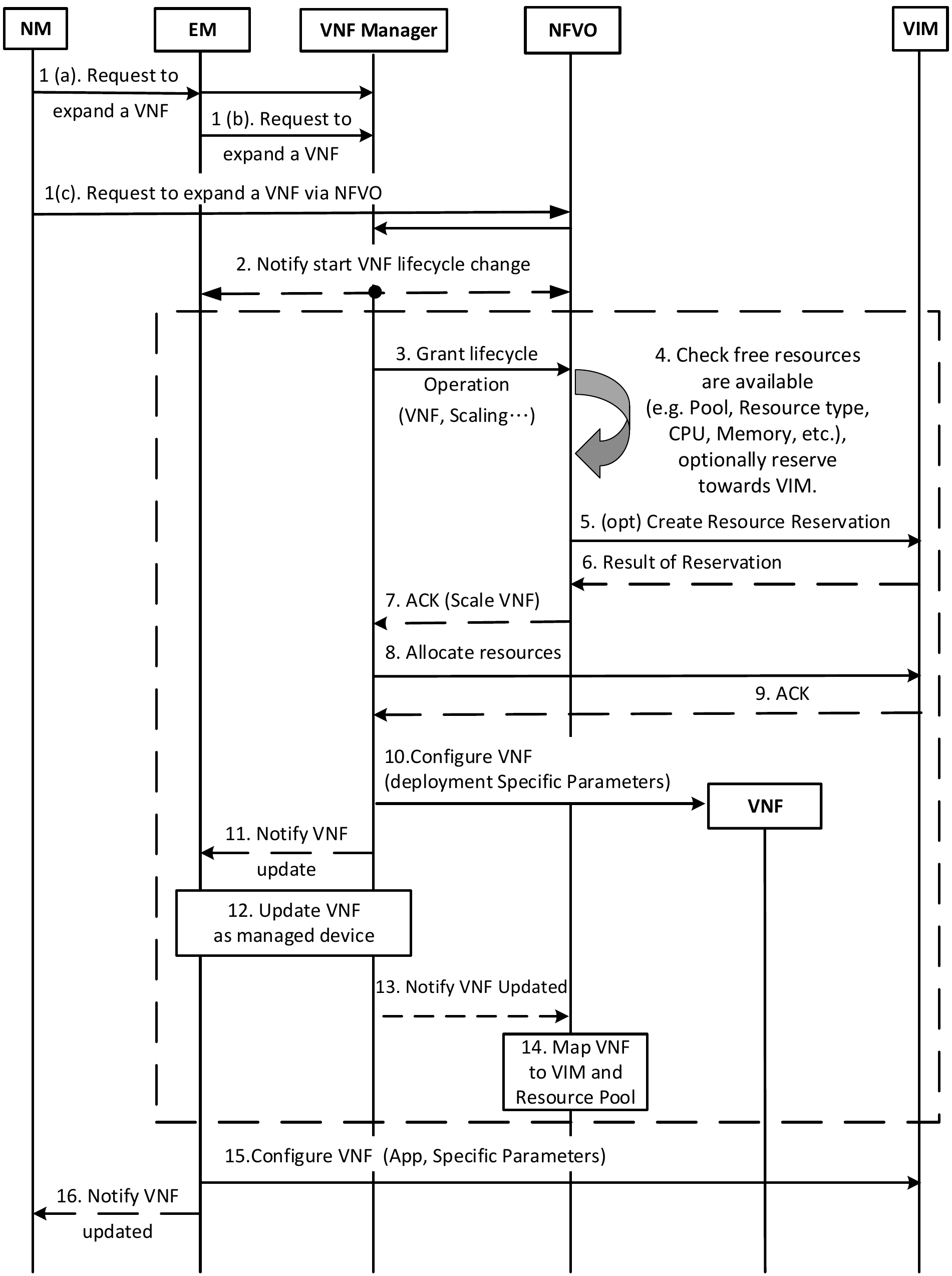}
	\caption{VNF instance expansion procedure triggered by NM/EM~\cite{V13.1.02015}.}
	\label{fig:VNF_expansion}
\end{figure}
\subsection{VNF Instance Scaling Procedures}

VNF manager allocates resources by using two scaling procedures: VNF instance expansion (scale-out and scale-up) procedure to add resources to a VNF, and VNF instance contraction (scale-in and scale-down) procedure to release the resources from a VNF.

\subsubsection{VNF instance expansion procedure}
Fig.~\ref{fig:VNF_expansion} illustrates the VNF instance expansion procedure. Here we briefly describe the flows. Please refer to~\cite{V13.1.02015} for details.
\begin{itemize}
	\item Step 1: NM/EM (via NFV Orchestrator, NFVO) sends capability expansion request to the VNF Manager, see 1(a), 1(b), and 1(c).
	\item Step 2: The VNF Manager sends lifecycle change notification to EM and NFVO indicating the start of the scaling operation.
	\item Step 3-14: The VNF Manager sends a request to the NFVO for the VNF expansion. The NFVO then checks whether free resources are available and send ACK/NACK to the VNF Manager for VNF expansion.
	\item Step 15: EM configures the VNF with application specific parameters.
	\item Step 16: EM notifies the newly updated and configured capacity to NM.
\end{itemize}

\subsubsection{VNF instance contraction procedure}
The idea of contraction procedure is similar to that of expansion procedure. Please refer to~\cite{V13.1.02015}.

\section{Related Work} \label{sec:Related_Work}

In cloud computing community, auto-scaling strategies have been studied intensively~\cite{xiao2013dynamic,jokhio2013prediction,roy2011efficient,tirado2011predictive, niu2012quality,huang2013migration,huang2014prediction,calheiros2014workload, islam2012empirical,bashar2013autonomic,bankole2013cloud,  shen2011cloudscale,khan2012workload,mitrani2013managing,Mitrani20111222, mitrani2013trading, hu2015power,phung2015multiserver}. To deal with the delay in VM setup, researchers have proposed various approaches to predict the VM load in order to boot VMs before existing VMs are overloaded. The approaches include Exponential weighted Moving Average (EMA)~\cite{xiao2013dynamic,jokhio2013prediction}, Auto-Regressive Moving Average (ARMA)~\cite{roy2011efficient,tirado2011predictive}, Auto-Regressive Integrated Moving Average (ARIMA)~\cite{niu2012quality,huang2013migration,huang2014prediction,calheiros2014workload}, machine learning~\cite{islam2012empirical, bashar2013autonomic,bankole2013cloud}, Markov model~\cite{shen2011cloudscale, khan2012workload}, and queueing model~\cite{mitrani2013managing,Mitrani20111222, mitrani2013trading, hu2015power,phung2015multiserver}.

The basic idea of EMA, ARMA, and ARIMA is moving average, where the most recent input data within a moving window are used to predict the next input data.
Specifically, in~\cite{xiao2013dynamic}, the authors proposed an EMA-based scheme to predict the CPU load. The scheme was implemented in Domain Name System (DNS) server and evaluation results showed that the capacities of servers are well utilized. The authors of~\cite{jokhio2013prediction} introduced a novel prediction-based dynamic resource allocation algorithm to scale video transcoding service in the cloud. They used a two-step prediction to predict the load, resulting in a reduced number of required VMs.

ARMA adds autoregressive (AR) into moving average. A resource allocation algorithm based on ARMA model was reported in~\cite{roy2011efficient}, where empirical results showed significant benefits both to cloud users and cloud service providers. In~\cite{tirado2011predictive}, the authors addressed a load forecasting model based on ARMA, which achieved around $6\%$ prediction error rate and saved up to $44\%$ hardware resources compared with random content-based distribution policy.

Unlike ARMA and ARIMA which differentiate input data, the authors of~\cite{niu2012quality} proposed a predictive and elastic cloud bandwidth auto-scaling system considering multiple data centers. This is the first work for linear scaling from multiple cloud service providers. The work of~\cite{huang2013migration} took the VM migration overhead into account when designing their auto-scaling scheme, where extensive experiments were conducted to demonstrate the performance. In~\cite{huang2014prediction}, a new problem of dynamic workload fluctuations of each VM and the resource conflict handling were addressed.  The authors further proposed an ARIMA-based server state predictor to adaptively allocate resource to VMs. Experiments showed that the state predictor achieved excellent prediction results. Another ARIMA-based workload prediction scheme was proposed in~\cite{calheiros2014workload}, where real traces of requests to web servers from the Wikimedia Foundation was used to evaluate its prediction accuracy. The results showed that the model achieved up to $91\%$ accuracy.

Machine learning approaches are also used for the design of cloud auto-scaling algorithms~\cite{islam2012empirical,bashar2013autonomic,bankole2013cloud}. The authors of~\cite{islam2012empirical} proposed a neural network and linear regression based auto-scaling algorithm. The author of~\cite{bashar2013autonomic} implemented a Bayesian Network based cloud auto-scaling algorithm. In~\cite{bankole2013cloud}, the authors evaluated three machine learning approaches: linear regression, neural network, and Support Vector Machine (SVM). Their results showed that SVM-based scheme outperforms the other two.

Markov model also has been widely used in cloud auto-scaling algorithms~\cite{shen2011cloudscale, khan2012workload}. The authors of~\cite{shen2011cloudscale} developed CloudScale, an automatic elastic resource scaling system for multiple cloud service providers, saving $8-10\%$ total energy consumption and $39-71\%$ workload energy consumption with little impact on the application performance. In~\cite{khan2012workload}, the authors proposed a novel multiple time series approach based on Hidden Markov Model (HMM). The technique well characterized the temporal correlations in the discovered VM clusters to predict variations of workload patterns.

However, \textit{the mechanisms in~\cite{xiao2013dynamic,jokhio2013prediction,roy2011efficient,tirado2011predictive, niu2012quality,huang2013migration,huang2014prediction,calheiros2014workload, islam2012empirical,bashar2013autonomic,bankole2013cloud,  shen2011cloudscale,khan2012workload} either ignore VM setup time or only consider virtualized resource itself while overlooking legacy (fixed) resources}. As aforementioned discussion, this is not practical for future 5G  cellular networks. Perhaps the closest models to ours were studied in~\cite{mitrani2013managing,Mitrani20111222, mitrani2013trading, hu2015power,phung2015multiserver} that both the capacities of fixed legacy network equipment and dynamic auto-scaling cloud servers are considered. The authors of~\cite{mitrani2013managing,Mitrani20111222} considered setup time without defections~\cite{mitrani2013managing} and with defections~\cite{Mitrani20111222}. Our recent work~\cite{hu2015power} relaxes the assumption in~\cite{mitrani2013managing,Mitrani20111222} such that after a setup time, all of the cloud servers in the block are active concurrently. We further consider a more realistic model that each server has an independent setup time. However, in~\cite{mitrani2013managing,Mitrani20111222, hu2015power}, all of the cloud servers were assumed as a whole block, which is not practical because each cloud server should be allowed to scale-out/in individually and dynamically. In~\cite{mitrani2013trading,phung2015multiserver}, it was relaxed to sub-blocks without considering all cloud servers as a whole block. However, either setup time is ignored~\cite{mitrani2013trading}, or legacy network capacity is not considered~\cite{phung2015multiserver}.

\section{Challenges and Contributions}
\label{sec:chall}
In this section, we summarize the challenges. How do we tackle the challenges and our contributions are also discussed.

\begin{itemize}
	\item The first challenge lies in the tradeoff between the operation cost and system performance, which is referred to as \textit{cost-performance tradeoff}. Keeping redundant powered-on VNF instances increases system performance and QoE. When a job arrives at the system, redundant powered-on VNF instances can serve the job immediately, which reduces job waiting time. On the other hand, the redundant powered-on VNF instances lead to extra operation cost. In this paper, we develop an analytical model to quantify the tradeoff.  Given the analytical model, operators can	quickly obtain the operation cost and system performance without real deployment to save cost and time.
	
	\item  The second challenge is how to count the capacity of legacy equipment  and how to choose a suitable value for the number of VNF instances to balance the cost-performance tradeoff. As aforementioned discussion, when the capacity of legacy equipment is counted as 10 VNF instances, a new power-up VNF instance increases $10\%$ system capability. Whereas, only $1\%$ capability is added if the legacy equipment are considered as 100 VNF instances. In addition, the power-up process is not instantaneous. During this setup process, the VNF instances consume power but cannot serve jobs. Based on our proposed analytical model, one can easily obtain the impacts of the capacity of legacy equipment and  the number of VNF instances to minimize the cost function.
	
	\item The third challenge is to propose a lightweight analytical model to quantify the tradeoff. In general, the computational cost to solve a Markov chain with $M$ states by a naive algorithm is $O(M^3)$. Thus, when $M$ is large, it is difficult to solve it in a short time.  In our proposed analytical model, we propose a novel recursive algorithm to reduce the computational complexity from $O(M^3)$ to $O(M)$, which is the same as the number of states of the Markov chain. The reduction is significant.
	
	\item Another challenge is how to adjust the auto-scaling algorithm in terms of different weighting factors for operation cost and system performance. Because different mobile operators may have different  management policies and operational interests, the weighting factors should be determined by a mobile operator. The adjustment of the algorithm according to different weighting factor is critical and non-trivial. Our proposed auto-scaling algorithm takes the weighting factors into consideration.
\end{itemize}

\section{Proposed Dynamic Auto-Scaling Algorithm (DASA)} \label{sec:Proposed_Algorithm}
In this section, we first introduce the system model. We then discuss the proposed DASA. The parameters used in the model are listed in Table~\ref{tab:ParameterSetting}.

\subsection{System Model}\label{ssec:System_Model}

We consider a 5G EPC comprised of both legacy network entities (e.g., MME, PCRF) and VNFs. A VNF, consisting of $k$ VNF instances, offers fine-grained on-demand network capabilities  to its corresponding legacy network entity. As shown in Fig.~\ref{fig:Queueing_model_special_case}, we assume that the capacity of the legacy network entity equals to $n_0$ numbers of VNF instances. The total capacity of the system is $N$ numbers of VNF instances ($N=n_0+k$), which can be adjusted adaptively depending on the number of $k$. We assume that $n_1=n_0+1$ and $n_i=n_{i-1}+1$ ($i=1, 2, \cdots k$). It should be noted that $n_k=N$.  User request arrives with rate $\lambda$. A VNF instance accepts one job at a time with service rate $\mu$. There is a limited First-Come-First-Served (FCFS) queue for those requests that have to wait to be processed. The legacy network equipment is always on while VNF instances will be added (or removed) according to the number of waiting jobs in the queue. It is worth to mention that the VNF instances need some setup time to be available so as to process waiting requests. During the setup time, the VNF instance consumes power\footnote{Here, the power may include server operation cost or VM cost charged by a cloud service provider.} but cannot serve jobs.

\begin{table}[t]
	\small
	\caption{List of Notations}
	\centering  \label{tab:ParameterSetting}
	\begin{tabular}{|c|p{7.5cm}|}
		\hline
				$N$           & the total capacity of the system                          \\ \hline
		$K$            & the maximum number of  jobs can be accommodated in the system \\ \hline
		$k$                & the number of VNF instances   \\ \hline
		$W$            & average response time per job \\ \hline
		$W_q$            & average response time in the queue per job \\ \hline
		$S$                & average VNF cost  \\ \hline
		$w_1$       & weighting factor for $W_q$   \\ \hline
		$w_2$        & weighting factor for $S$     \\ \hline
		$n_{0}$      & the capacity of a legacy network entity                      \\ \hline
		$U_{i}$       & the up threshold to control the VNF instances           \\ \hline
		$D_{i}$       & the down threshold to control the VNF instances  \\ \hline
		$\lambda$ & job arrival rate                                \\ \hline
		$\mu$       & service rate for each VNF instance                   \\ \hline
		$\alpha$     & setup rate for each VNF instance              \\ \hline
	\end{tabular}
\end{table}
\begin{figure}[t]
	\centering
	\includegraphics[width=8.7cm]{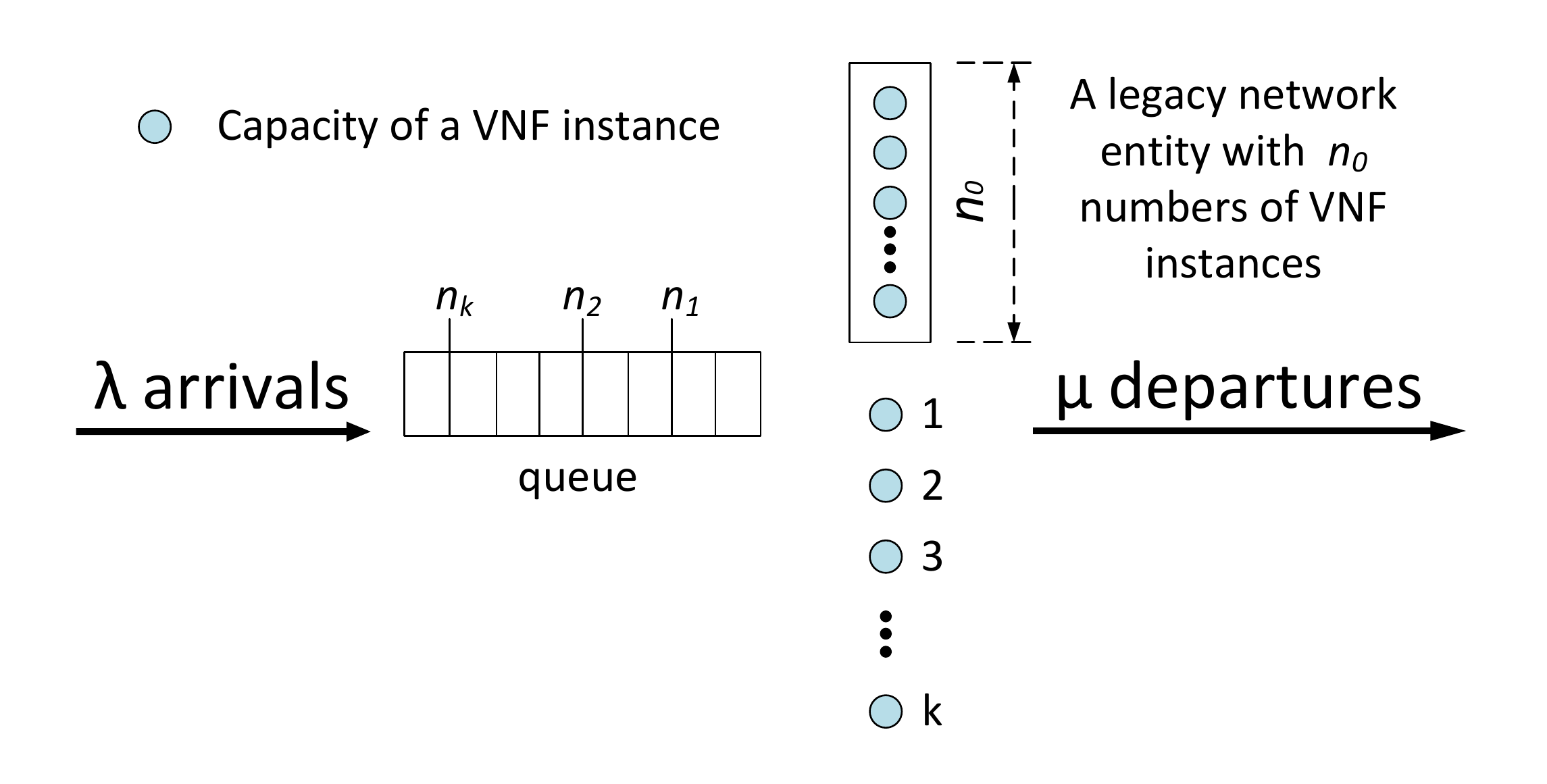}
	\caption{A simplified queueing model for our system.}
	\label{fig:Queueing_model_special_case}
	\vspace{-3mm}
\end{figure}

\subsection{Cost Function}
\label{ssec:cost}

Our goal is to design the \textit{best} auto-scaling strategy to minimize operation cost while providing acceptable levels of performance.  We use two thresholds, \textit{up} and \textit{down}, or $U_i$ and $D_i$, to denote the control of the VNF instances, where $i=1,2,\cdots, k$.
\begin{itemize}
	\item \textit{$U_i$, power up the $i$-th VNF instance:}  If the $i$-th VNF instance is turned off and the number of requests in the system increases from $U_i-1$ to $U_i$, the VNF instance is powered up after a setup time to support the system.  During the setup time, a VNF instance cannot serve user requests, but consumes power (or money for renting cloud services). Here, we specify $U_i=n_i$.
	\item \textit{$D_i$, power down the $i$-th VNF instance:}  If the $i$-th VNF instance is operative, and the number of requests in the system drops from $D_i+1$ to $D_i$, the VNF instance is powered down instantaneously. Here, we define $D_i = n_{i-1}$.
\end{itemize}
While powering up/down the VNF instance, the question is how many VNF instances we need such that the cost is minimized while the required level of performance is also met. Here, the cost function $C$ could be evaluated by two metrics: the average response time in the queue per request, $W_q$, and the average cost of VNF instances, $S$.  The mathematical formulation of the cost function\footnote{A variance of the cost function with more parameters can be found in Appendix.} can be written as:
\begin{equation}\label{eq:optimal_short}
\begin{aligned}
& \underset{}{\text{minimize}}
& & C =  w_1W_q+ w_2S\; ,\\
& \text{subject to}
& & 0<W_q<W_q'\; ,
\end{aligned}
\end{equation}
where  $W_q'$ is the upper bound of $W_q$, which can be determined by mobile operators according to their business policies. The coefficients of $w_1$ and $w_2$ denote the weighting factors for $W_q$ and $S$, respectively. Increasing $w_1$ (or $w_2$) emphasizes more on $W_q$ (or $S$). Here, we do not specify either $w_1$ or $w_2$ because such a value should be determined by mobile operators and should take management policies into consideration. An algorithm for finding the optimal solution is introduced in Section~\ref{ssec:algo} if $w_1$ and $w_2$ are specified.

\subsection{Derivation of Cost Function}
\label{ssec:deri}

We model the system as a queueing model with $N=n_0+k$ servers divided into two blocks: fixed block and dynamic block. The $n_0$ servers in fixed block are always on (refer to the capacity of legacy equipment).  The dynamic block denotes VNF instances in which $k$ servers are in either BUSY, OFF, or SETUP state. The queueing model has a capacity of $K$, i.e., the maximum number of jobs can be accommodated in the system is $K$. Job arrivals follow Poisson distribution with rate $\lambda$. A VNF instance, which is referred to as a server in the queueing system, accepts one job at a time, and its service rate follows exponential distribution with rate $\mu$. There is a limited FCFS queue for those jobs that have to wait to be processed.

In dynamic block, a server is turned off immediately if it has no job to serve. Upon arrival of a job, an OFF server is turned on if the job is placed in the buffer. However, a server needs some setup time to be active in order to serve waiting jobs. We assume that the setup time follows exponential distribution with mean $1/\alpha$. Let $j$ denotes the number of customers in the system and $i$ denotes the number of active servers in the dynamic block. The number of servers in SETUP status is $\min~(j-n_i, N-n_i)$.  We assume that waiting jobs are served according to FCFS. We call this model an $M$/$M$/$N$/$K$/$setup$ queue.

Here, we present a recursive algorithm to calculate the joint stationary distribution.
Let $C(t)$ and $L(t)$ denote the number of active servers in the dynamic block and the number of customers in the whole system, respectively. It is easy to see that $\{ X(t) = (C(t),L(t)); t \geq 0\}$ forms a Markov chain on the state space:
\begin{align}
\nonumber \mathcal{S} =  &   \{(i,j); 1 \leq i \leq k, j  = n_i,n_i+1,\dots,K-1,K\}\\
& \cup \{(0,j); j = 0,1,\dots,K-1,K\}.
\end{align}

\begin{figure}[t]
	\begin{center}
		\includegraphics[width=8.6cm]{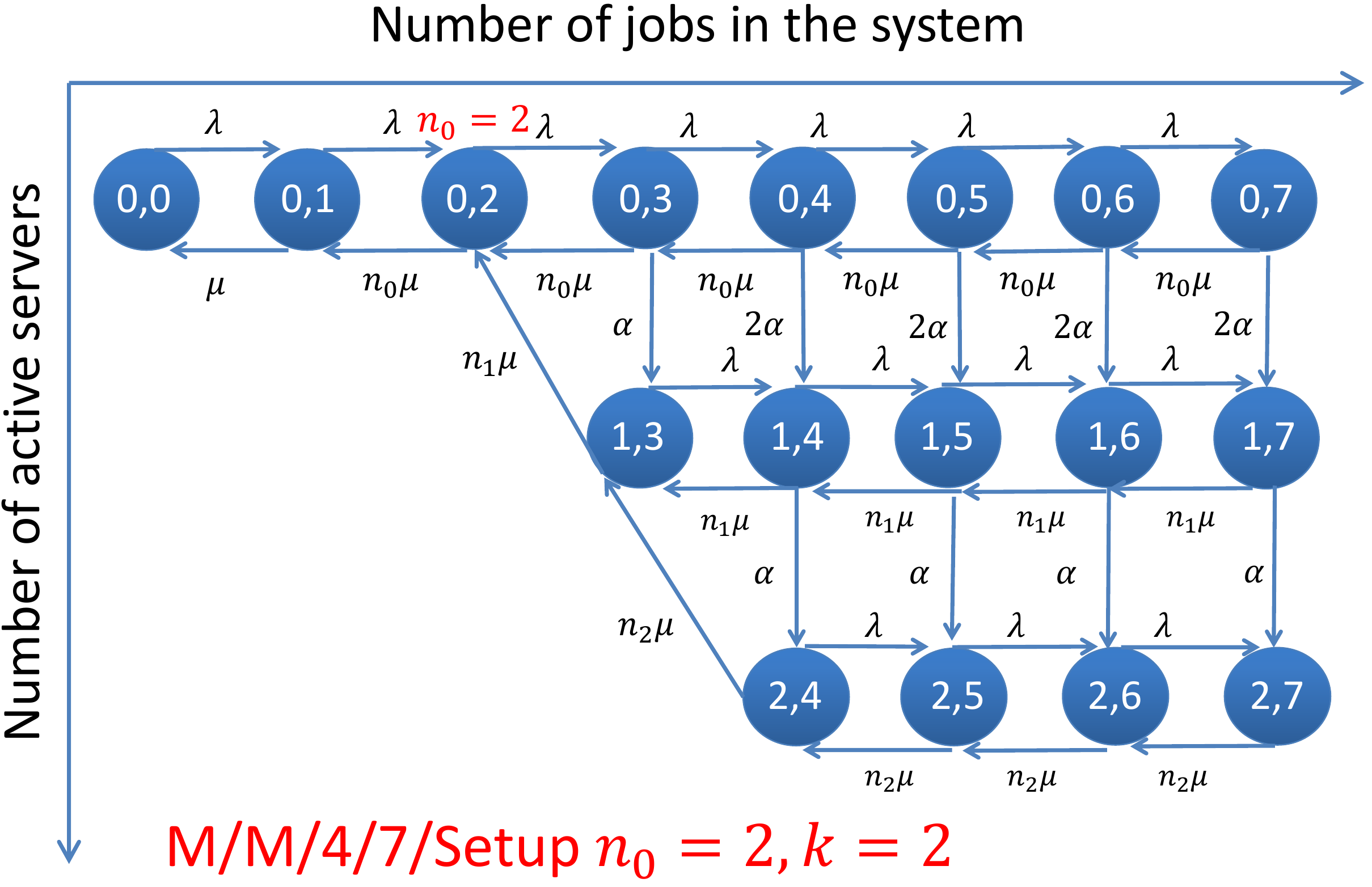}
		\caption{Transition among states ($N=4$, $n_0 = 2$, $k=2$, and $K=7$).}
		\label{fig:Markov_Chain_4_7}
	\end{center}
\end{figure}

Fig.~\ref{fig:Markov_Chain_4_7} shows the transition among states for the case  where $N=4$, $n_0 = 2$, $k=2$, and $K=7$.  Let
\begin{align}
\pi_{i,j} = \lim_{t \to \infty} {\rm P} (C(t) = i, L(t) = j), (i,j) \in \mathcal{S}
\end{align}
denote the joint stationary distribution of $\{X(t) \}$.
Here, we derive a recursion to calculate the joint stationary distribution $\pi_{i,j}$, $(i,j) \in \mathcal{S}$.

First, we consider a recursion for $\pi_{0,j}$ ($j = 0,1,\dots,K$). The balance equations for states with $i=0$ are given as follows:
\begin{align}
\lambda \pi_{0,j-1} \textit{} &=  j \mu \pi_{0,j}, \quad \mathrm{for}\;\; j = 0,1,\dots,n_0,  \\
\nonumber \lambda \pi_{0,j-1}  + n_0 \mu \pi_{0,j+1} &=  (\lambda + n_0 \mu\\
&+ \min (j-n_0,N-n_0) \alpha    ) \pi_{0,j} ,\\
\nonumber &\quad \mathrm{for}\;\; j = n_0,n_0 + 1, \dots, K-1, \\
\lambda \pi_{0,K-1}& =   (n_0 \mu + (N-n_0) \alpha) \pi_{0,K},
\end{align}
leading to:
\begin{align}\label{pi0j:eq}
\pi_{0,j} = b^{(0)}_j \pi_{0,j-1}, \qquad j=1,2,\dots, K.
\end{align}
The sequence $\{ b^{(0)}_j; j = 1,2,\dots,K \}$ is given as follows:
\begin{align}
b^{(0)}_j = \frac{\lambda}{j\mu}, \qquad j = 1,2,\dots,n_0,
\end{align}
and
\begin{align}
b^{(0)}_j = \frac{\lambda}{\lambda + n_0 \mu + \min (j-n_0,N-n_0) \alpha - n_0 \mu b^{(0)}_{j+1}},\\
\nonumber \qquad j = K-1,K-2,\dots,n_0 + 1,
\end{align}where
\[
b^{(0)}_K = \frac{\lambda}{ n_0 \mu + (N-n_0) \alpha }.
\]
Furthermore, it should be noted that $\pi_{1,n_1}$ is calculated using the local balance equation in and out the set $\{(0,j); j = 0,1,\dots,K\}$ as follows:

\begin{align}\label{expression_p11:eq}
n_1 \mu \pi_{1,n_1} = \sum_{j=n_1}^K \min (j,N-n_0) \alpha \pi_{0,j}.
\end{align}

\begin{remark}
	We have expressed $\pi_{0,j}$ ($j = 1,2,\dots,K$) and $\pi_{1,n_1}$ in terms of $\pi_{0,0}$.
\end{remark}

We consider the general case for $\pi_{i,j}$ where $1 \leq i \leq k-1$. Lemma~\ref{lemma:pi_ij} below shows that for a fixed $i = 1,2,\dots,k-1$, $\pi_{i,j}$ can be expressed in terms of $\pi_{i,j-1}$ ($j = n_i+1,n_i+2,\dots,K$). As a result, $\pi_{i,j}$ ($j = n_i+1,n_i+2,\dots,K$) is expressed in terms of $\pi_{i,n_i}$.
%
%
\begin{lem}\label{lemma:pi_ij}
	We have:
	\begin{align}
	\pi_{i,j} = a^{(i)}_j + b^{(i)}_j \pi_{i,j-1}, \qquad j = n_i + 1, n_i +2, \dots, K-1,K,
	\end{align}
	where
	\begin{align}\label{backward:i}
	a^{(i)}_j  & =  \frac{ n_i\mu a^{(i)}_{j+1} +  \min(N-n_{i-1},j-n_{i-1}) \alpha \pi_{i-1,j}  }{ \lambda + \min (N-n_i,j-n_i) \alpha + n_i \mu - n_i \mu b^{(i)}_{j+1}  }, \\
	\label{backward:ib}
	b^{(i)}_j  & =  \frac{\lambda}{ \lambda + \min (N-n_i,j-n_i) \alpha + n_i \mu - n_i \mu b^{(i)}_{j+1}  },
	\end{align}
	and
	\begin{align}
	a^{(i)}_K = \frac{(N-n_{i-1}) \alpha \pi_{i-1,K}}{(N-n_i) \alpha + n_i \mu}, \qquad b^{(i)}_K = \frac{\lambda }{(N-n_i) \alpha + n_i \mu}.
	\end{align}
\end{lem}
\begin{proof}
	The balance equation for state $(i,K)$ is given as follows:
	\begin{align}
	((N-n_i) \alpha + n_i \mu) \pi_{i,K} = \lambda \pi_{i,K-1} + (N-n_{i-1}) \alpha \pi_{i-1,K},
	\end{align}
	It leads to the fact that Lemma~\ref{lemma:pi_ij} is true for $j=K$.
	Assuming that:
	\begin{align}
	\pi_{i,j+1} = a^{(i)}_{j+1} + b^{(i)}_{j+1} \pi_{i,j}, \qquad j = n_i + 1, n_i +2, \dots, K-1.
	\end{align}
	Substituting this into the next balance equation:
	\begin{align}
	\nonumber 	(\lambda + \min (N-n_i,j-n_i) \alpha + n_i \mu ) \pi_{i,j} =  \lambda \pi_{i,j-1} \\
	+ n_i \mu \pi_{i,j+1} + \min(N-n_{i-1},j-n_{i-1}) \alpha \pi_{i-1,j}, \\
	\nonumber  \qquad j = K-1,K-2,\dots,n_i+1,
	\end{align}
	we obtain:
	\begin{align}
	\pi_{i,j} = a^{(i)}_{j} + b^{(i)}_{j} \pi_{i,j-1}.
	\end{align}
\end{proof}

\begin{remark}
	In Corollary~\ref{theorem:i} below, we will show that $a^{(i)}_j$ and $b^{(i)}_j$ are positive. Thus, the recursive algorithm is stable because it manipulates only positive numbers. Furthermore, we can also prove that $b^{(i)}_j$ is bounded from above.  Although we cannot obtain an explicit upper bound for $a^{(i)}_j$, from numerical experiments, we observe that $a^{(i)}_j$ is not so large.  One reason may be that the coefficient of $a^{(i)}_{j+1}$ in (\ref{aij_ineq}) is less than 1. These upper bounds are the rationale for the stability in our recursive algorithm because we deal with numbers that are not too large so that overflow is avoided.
\end{remark}


\begin{coro}\label{theorem:i}
	We have the following bound.
	\begin{align}
	\label{aij_ineq}
	& 0 < a^{(i)}_j < \frac{ n_i\mu a^{(i)}_{j+1} +  \min(N-n_{i-1},j-n_{i-1}) \alpha \pi_{i-1,j}  }{n_i \mu  + \min(j-i,N-n_i) \alpha}, \\
	\label{bij_ineq}
	& 0 < b^{(i)}_j  < \frac{\lambda}{n_i \mu  + \min(j-i,N-n_i) \alpha},
	\end{align}
	for $j = n_i+1,n_i+2,\dots,K, i = 1,2,\dots,k-1$.
\end{coro}
\begin{proof}
	We prove it by using mathematical induction. It is clear that Corollary~\ref{theorem:i} is true for $j=K$. Assuming that Corollary~\ref{theorem:i} is true for $j+1$, i.e.,
	\begin{align}
	a^{(i)}_{j+1} > 0, \qquad 0 < b^{(i)}_{j+1}  < \frac{\lambda}{n_i\mu  + \min(j+1-n_i,N-n_i) \alpha},
	\end{align}
	for $j = n_i+1,n_i+2,\dots,K-1, i = 1,2,\dots,k-1$.
	It follows from the second inequality that $n_i\mu b^{(i)}_{j+1} < \lambda$. This together with Equations~(\ref{backward:i}) and (\ref{backward:ib}) yield the desired result.
\end{proof}

It should be noted that $\pi_{i+1,n_{i+1}}$ is calculated using the following local balance equation in and out the set of states:
\begin{align}
\{(k,j) \in \mathcal{S}; k=0,1,\dots,i \}
\end{align}
as follows:
\begin{align}\label{pii+1:eq}
n_{i+1} \mu \pi_{i+1,n_{i+1}} = \sum_{j=n_i+1}^K \min(j-n_i,N-n_i) \alpha \pi_{i,j}.
\end{align}
\begin{remark}
	We have expressed $\pi_{i,j}$ ($i=0,1,\dots,k-1, j = n_i,n_i+1,\dots,K$) and $\pi_{k,n_k}$ in terms of $\pi_{0,0}$.
\end{remark}

Finally, we consider the case $i = k$. The balance equation for state $(k,j)$ ($j = k,k+1,\dots,K$) leads to Lemma~\ref{lemma3}.
\begin{lem}
	\label{lemma3}
	We have:
	\begin{align}
	\pi_{k,j} = a^{(k)}_j + b^{(k)}_j \pi_{k,j-1}, \qquad j = n_k + 1, n_k +2, \dots, K,
	\end{align}
	where
	\begin{align}\label{backward:c}
	a^{(k)}_j  &=  \frac{n_k \mu a^{(k)}_{j+1} + (N-n_{k-1})\alpha \pi_{k-1,j}}{\lambda + n_k\mu - n_k \mu b^{(k)}_{j+1}}, \quad\\
	\nonumber j &= K-1, K-2, \dots, n_k+1,\\
	\label{backward:cb}
	b^{(k)}_j  &=  \frac{\lambda}{\lambda + n_k\mu - n_k \mu b^{(k)}_{j+1}}, \quad\\
	\nonumber j &= K-1, K-2, \dots, n_k+1,
	\end{align}
	and
	\begin{align}
	a^{(k)}_K = \frac{\alpha \pi_{k-1,K}}{n_k \mu}, \qquad b^{(k)}_K = \frac{\lambda}{n_k \mu}.
	\end{align}
\end{lem}
\begin{proof}
	The global balance equation in state $(k,K)$ is given by:
	\begin{align}
	n_k \mu \pi_{k,K} = (N-n_{k-1}) \alpha \pi_{k-1,K} + \lambda \pi_{k,K-1},
	\end{align}
	leading to:
	\begin{align}
	\pi_{k,K} = a^{(k)}_K + b^{(k)}_K \pi_{k,K-1}.
	\end{align}
	Assuming that $\pi_{k,j+1} = a^{(k)}_{j+1} + b^{(k)}_{j+1} \pi_{k,j}$,  it follows from this formula and the global balance equation in state $(k,j)$:
	\begin{align}
	\nonumber 	(\lambda + n_k \mu) \pi_{k,j} =& \lambda \pi_{k,j-1} + n_k \mu \pi_{n_k,j+1}\\
	&+ (N-n_{k-1}) \alpha \pi_{k-1,j},\\
	\nonumber		j = &n_k+1,n_k+2,\dots,K-1,
	\end{align}that $\pi_{k,j} = a^{(k)}_j + b^{(k)}_j \pi_{k,j-1}$ for $j=n_k+1,n_k+2,\dots,K$.
\end{proof}
\begin{coro}\label{ab_cj:thm}
	We have the following bound.
	\begin{align}
	a^{(k)}_j >0, \qquad 0 < b^{(k)}_j  < \frac{\lambda}{n_k \mu}, \qquad \\
	\nonumber j = n_k+1,n_k+2,\dots,K-1.
	\end{align}
\end{coro}
\begin{proof}
	We also prove it by using mathematical induction. It is clear that Corollary~\ref{ab_cj:thm} is true for $j=K$. Assuming that Corollary~\ref{ab_cj:thm} is true for $j+1$, i.e.,
	\begin{align}
	a^{(k)}_{j+1} >0, \qquad 0 < b^{(k)}_{j+1}  < \frac{\lambda}{n_k \mu}, \\
	\nonumber		\qquad j = n_k+1,n_k+2,\dots,K-1.
	\end{align}
	
	It follows from the second inequality that $n_k \mu b^{(k)}_{j+1} < \lambda$. This together with Equations~(\ref{backward:c}) and (\ref{backward:cb}) yield the desired result.
\end{proof}

We have expressed all of the probabilities $\pi_{i,j}$ ($(i,j) \in \mathcal{S}$) in terms of $\pi_{0,0}$ which is uniquely determined by the normalizing condition.
\begin{align}
\sum_{(i,j) \in \mathcal{S}} \pi_{i,j} = 1.
\end{align}

\begin{remark}
	In summary, we can calculate all probabilities $\pi_{i,j}$ ($(i,j) \in \mathcal{S}$) in the following order.
	First, we set $\pi_{0,0} = 1$. We then calculate all the probabilities $\pi_{0,j}$ ($j = 1,2,\dots,K$) using Equation (\ref{pi0j:eq}). Next, $\pi_{1,n_1}$ is calculated using (\ref{expression_p11:eq}). After that, we apply Lemma~\ref{lemma:pi_ij} and (\ref{pii+1:eq}) repeatedly for $i=1,2,\dots,k-1$. At this point, $\pi_{i,j}$ for $i=0,1,\dots,k-1$ and $\pi_{k,n_k}$ are obtained. Furthermore, we use  Lemma~\ref{lemma3} in order to obtain $\pi_{k,j}$ for $j = n_k+1,n_k+2,\dots,K$. Finally, we divide all $\pi_{i,j}$ ($(i,j) \in \mathcal{S}$) by $\sum_{(i,j) \in \mathcal{S}} \pi_{i,j}$ in order to get the stationary distribution.
\end{remark}

Let $L$ denote the mean number of jobs in the system. We have:
\begin{align}\label{eq:L}
L =\sum_{(i,j) \in \mathcal{S}} \pi_{i,j} j = \sum_{i=0}^{n_0-1} \pi_{0,j} j + \sum_{i=0}^k \sum_{j=n_i}^K \pi_{i,j} j.
\end{align}

It follows from Little's law that:
\begin{align}\label{eq:W}
W = \frac{{\rm E}[L]}{\lambda (1-P_b)}=\frac{\sum_{i=0}^{n_0-1} \pi_{0,j} j + \sum_{i=0}^k \sum_{j=n_i}^K \pi_{i,j} j}{\lambda (1-\sum_{i=0}^k \pi_{i,K})}.
\end{align}

Therefore, we obtain:
\begin{align}
W_q = W-\frac{1}{\mu}. \label{eq:W_q}
\end{align}

Let $P_b$ denote the blocking probability. We have:
\begin{align}\label{eq:P_b}
P_b = \sum_{i=0}^k \pi_{i,K}.
\end{align}

The mean number of VNF instances is given by:
\begin{align}
S = \sum_{(i,j) \in \mathcal{S}} \pi_{i,j} (n_i-n_0) + \sum_{i=0}^k \sum_{j = n_i}^K \pi_{i,j} \min(j-n_i,N-n_i), \label{eq:S}
\end{align}
where the first term is the number of VNF instances that are active already while the second term is the mean number of VNF instances in setup mode.

\textbf{Summary of the derivation:} In this section, we have developed a mathematical model to derive the metrics $W_q$ and $S$ in the cost function~(\ref{eq:optimal_short}), where $W_q$ and $S$ are shown in Equations~(\ref{eq:W_q}) and  (\ref{eq:S}), respectively.  Given the closed forms, one can easily find the optimal $\tau \in \{k, n_0, \mu, K, \alpha\}$ to balance the cost function $C$ if the rest of parameters and $\lambda$ are given. The reason is that $W_q$ and $S$ are the functions of those parameters. 	
We have:
	\begin{equation}\label{eq:optimal}
	\begin{aligned}
	& \underset{\tau}{\text{arg min}}
	& & C =  w_1W_q+ w_2S\; ,\\
	& \text{subject to}
	& & 0<W_q<W_q'\; .
	\end{aligned}
	\end{equation}
We can find the local maximum/minimum of the cost function at point $\tau$ when $C'=0$ is satisfied. In addition, $C$ at point $\tau$ has the local minimum if $C''>0$. The optimal $\tau$ is then obtained.
In next section, we will use $\tau=k$ as an example to show how to use the derived metrics to decide the optimal $k$ and the number of VNF instances according to their weighting factors. Please note that we can also set $\tau$ as other parameters, such as $ n_0, \mu, K, \alpha$, which can also be easily applied to the algorithm introduced in the next section.

Moreover, other metrics such as  $P_b$, $W$, and $L$ are also given in Equations~(\ref{eq:P_b}), (\ref{eq:W}), and (\ref{eq:L}), which can be used for variants of the cost function (See Appendix).
The derivations in this section and Appendix are generic models which can be easily extended to any number of metrics. The changes in the number of metrics will not alter our
analysis although they may affect the optimal policies set by operators.

It is also worth to mention that we have solved a system of $n_0 + \sum_{i=0}^k (K-n_i)$ = $O(k K)$ unknown variables. The computational complexity by a conventional method is $O(k^3 K^3)$. It is easy to see that the computational complexity of our recursive algorithm is only $O(k K)$. Furthermore, our algorithm is numerically stable since it manipulates positive numbers.

\subsection{Algorithm for Deciding $k$}
\label{ssec:algo}
\begin{algorithm}[t]
	\caption{Selecting the optimal $k$}
	\begin{algorithmic}[1] \label{algo:opti}
		\renewcommand{\algorithmicrequire}{\textbf{Input:}}
		\REQUIRE $\overline S$, $\overline W_q$, $\delta$, $K$
		\renewcommand{\algorithmicensure}{\textbf{Output:}}
		\ENSURE $k_{op}$
		\STATE Initialize $k$ as 0
		\WHILE{$k \leq K - n_0$}
		\STATE $S^\prime = S / \overline S$
		\STATE $W_q^\prime = W_q / \overline W_q$
		\IF {$S^\prime/W_q^\prime < \delta $}
		\STATE {$k = k + 1$}
		\ELSE
		\STATE {return $k$}
		\ENDIF
		\ENDWHILE
	\end{algorithmic}
\end{algorithm}

Given the analytical model above, one can quickly obtain the operation cost and system performance and design optimal strategies without real deployment to save cost and time. Without our model, it is difficult to obtain the results in a short time. For instance, even in our simplified simulation settings (few arrival rate) in Section~\ref{sec:Simulation_Results}, it is still very time-consuming to get simulation results, e.g., tens of hours per simulation.

We propose Algorithm~\ref{algo:opti} for operators to specify $k$ based on the weighting factors. For ease of understanding, we use two weighting factors and two parameters, i.e., the cost function in~(\ref{eq:optimal_short}).   The Algorithm~\ref{algo:opti} takes input $\overline S$, $\overline W_q$, $\delta=\frac{w_2}{w_1}$  and outputs the optimal value $k_{op}$. We denote $\overline S$ and $\overline W_q$ as the maximum values of $S$ and $W_q$ in the system. Note that $\overline S$ and $ \overline W_q$ are the constraints set by the operators.
Initially, we set $k$ to 0 and  $k$ is bound by $K-n_0$.
As $k$ starts from 0, the ratio of $S^\prime / W_q^\prime$ increases in every loop accordingly. The loop does not stop until it finds the lowest $k$ value for  $k_{op}$.

\begin{figure}[t]
	\centering
	\includegraphics[width=8cm]{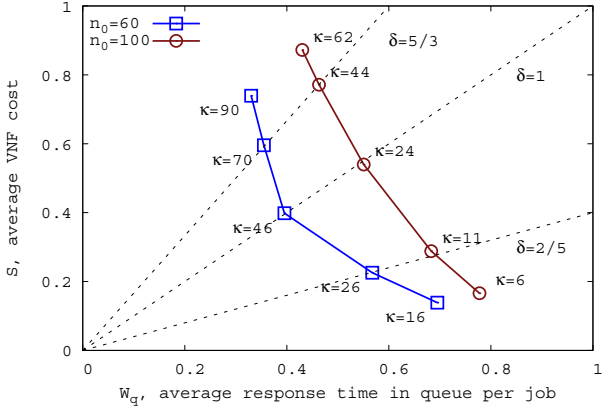}
	\caption{Selection of the optimal $k$ for a given $\delta$.}
	\label{fig:roc_k}
	\vspace{-5mm}
\end{figure}

Fig.~\ref{fig:roc_k} illustrates a graphical plot of $S$ and $W_q$ to demonstrate how to get $k$ for different settings. The three dotted black lines are with different weighting factors ($\delta=\frac{2}{5}, 1, \frac{5}{3}$). Each point in the blue curve is depicted from paired $S$ and $W_q$ associated with different $k$ when $n_0=60$. The blue curve is then plotted with different values of $k$. Similarly, we can plot the brown curve when $n_0=100$. The intersections of the dotted lines and the curves are then the optimal values of $k$ with the chosen parameters. Take the brown curve ($n_0=100$) as an example, the optimal values of $k$ are $44$, $24$, and $11$ for $\delta=\frac{5}{3}$, $\delta=1$, and $\delta=\frac{2}{5}$, respectively. Please note we simply use Fig.~\ref{fig:roc_k} to explain the idea. Operators can use Algorithm~\ref{algo:opti} to get the optimal value of $k$.

\section{Numerical Results} \label{sec:Simulation_Results}

In this section, we show the numerical results. The analytical results in Section~\ref{sec:Proposed_Algorithm} are validated by extensive simulations by using ns-2, version 2.35~\cite{ns2}. In simulation, we use real measurement results for parameter configuration:  $\lambda$ by Facebook data center traffic~\cite{lambda}, $\mu$ by the base service rate of a Amazon EC2 VM~\cite{gilani2015application}, and $\alpha$ by the average VM startup time~\cite{alpha}.  If not further specified, the following parameters are set as the default values for performance comparison: $n_0 = 110$, $\mu = 1$, $\alpha = 0.005$, $K = 250$, $\lambda = 50\thicksim250$ (see Table 1 for details). There were $15\thicksim 750$ million job requests generated during the simulations. Please note those parameters can be replaced by other values. We simply use them to validate our mathematical model and demonstrate the  numerical results.

Figs.~\ref{fig:Impacts_S_expo}-\ref{fig:Impacts_Wq_expo} illustrate both the simulation and analytical results in terms of average VNF cost $S$ and average response time in a queue per job $W_q$, respectively. In the figures, the \textit{lines} denote analytical results and the \textit{points} represent simulation results.  Each simulation result in the figures is the mean value of the results in 300,000 seconds with 95\% confidence level. In the following sections, we show the impacts of $\lambda$, $k$, $K$, $n_0$, $\alpha$ on the performance metrics $S$ and $W_q$, respectively.
\begin{figure*}
	\centering
	\begin{subfigure}[b]{0.48\textwidth}
		\includegraphics[width=8.5cm]{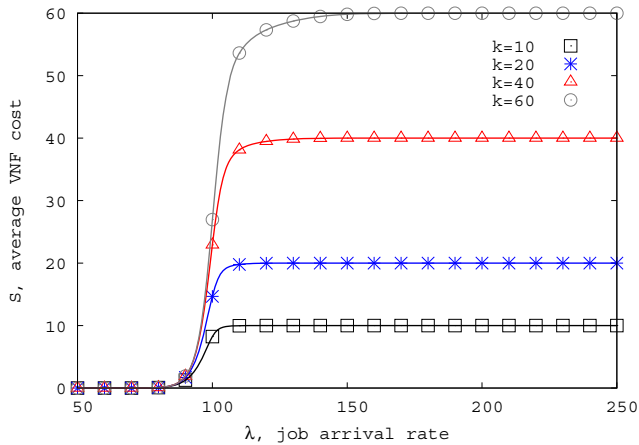}
		\caption{Impacts of $k$ on $S$ ($n_0 = 100$).}
		\label{fig:S_k}
\vspace{5mm}
	\end{subfigure}
	\begin{subfigure}[b]{0.48\textwidth}
		\includegraphics[width=8.5cm]{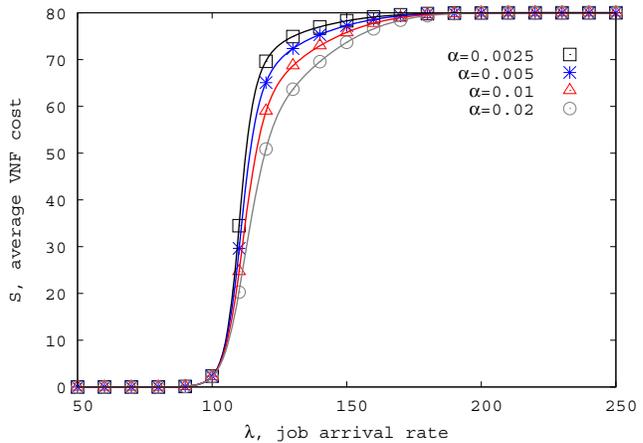}
		\caption{Impacts of $\alpha$ on $S$ ($k = 80$).}
		\label{fig:S_alpha}
\vspace{5mm}
	\end{subfigure}
	\begin{subfigure}[b]{0.48\textwidth}
		\includegraphics[width=8.5cm]{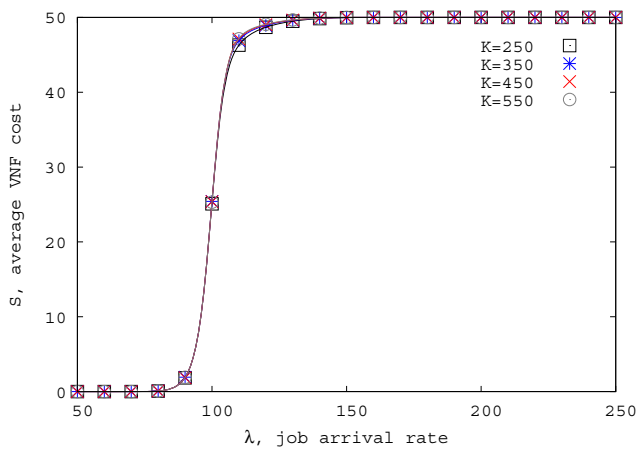}
		\caption{Impacts of $K$ on $S$ ($k = 50$).}
		\label{fig:S_K}
	\end{subfigure}
	\begin{subfigure}[b]{0.48\textwidth}
		\includegraphics[width=8.5cm]{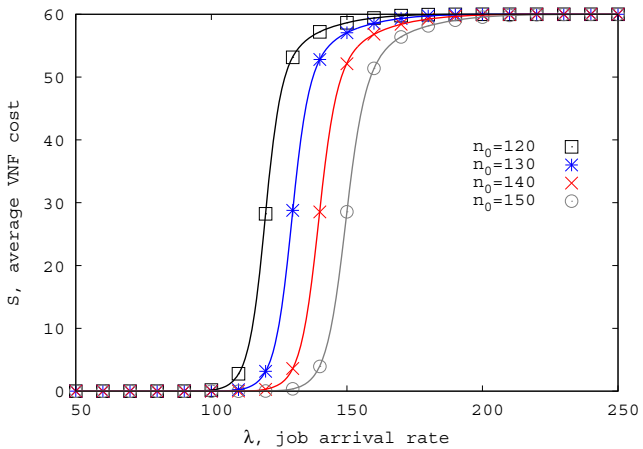}
		\caption{Impacts of $n_0$ on $S$ ($k = 60$).}
		\label{fig:S_n0}
	\end{subfigure}
	\caption{Impacts on $S$ while $1/\mu$, $1/\lambda$, and $1/\alpha$ are exponential distribution.}
	\label{fig:Impacts_S_expo}
\end{figure*}

\begin{figure*}
	\centering
	\begin{subfigure}[b]{0.48\textwidth}
		\includegraphics[width=8.5cm]{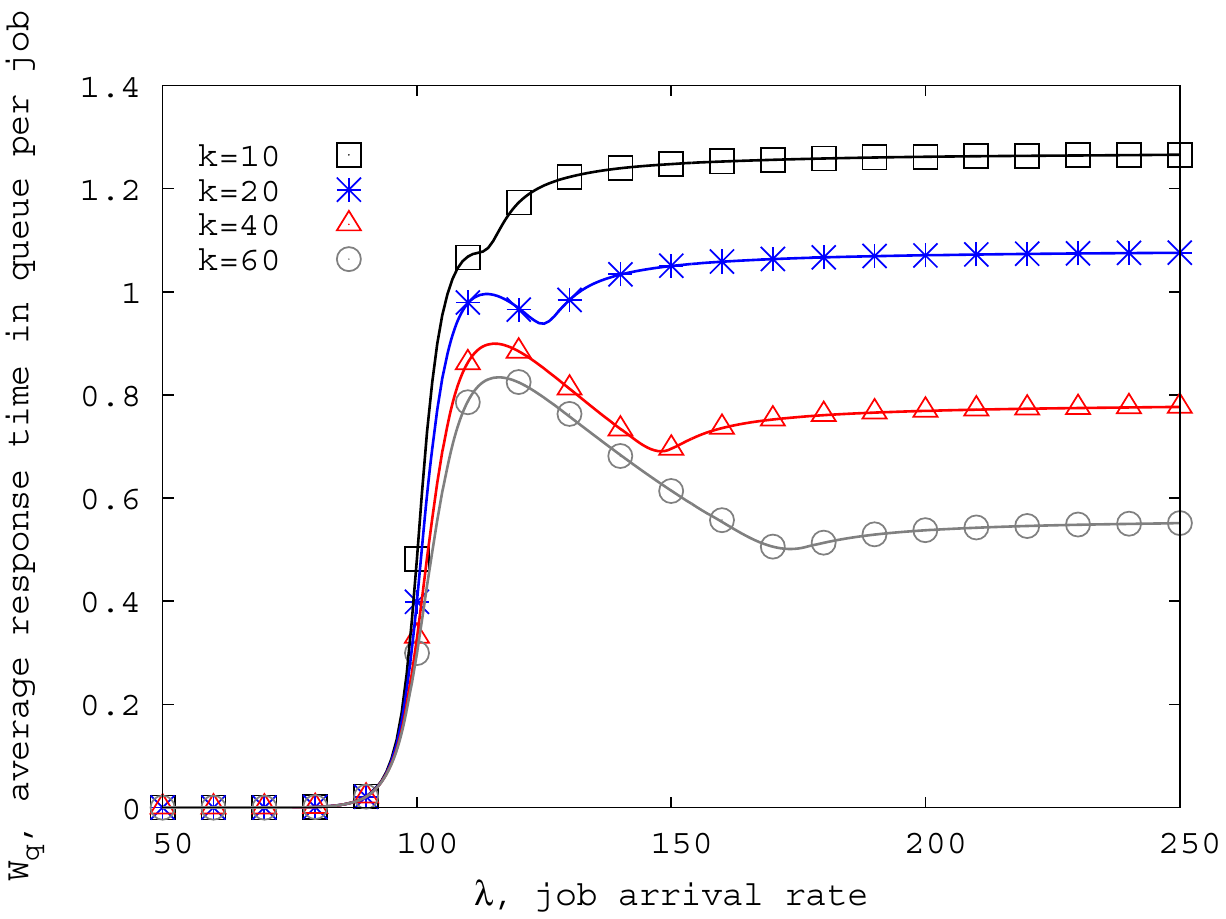}
		\caption{Impacts of $k$ on $W_q$ ($n_0 = 100$).}
		\label{fig:Wq_k}
\vspace{4mm}
	\end{subfigure}
	\begin{subfigure}[b]{0.48\textwidth}
		\includegraphics[width=8.5cm]{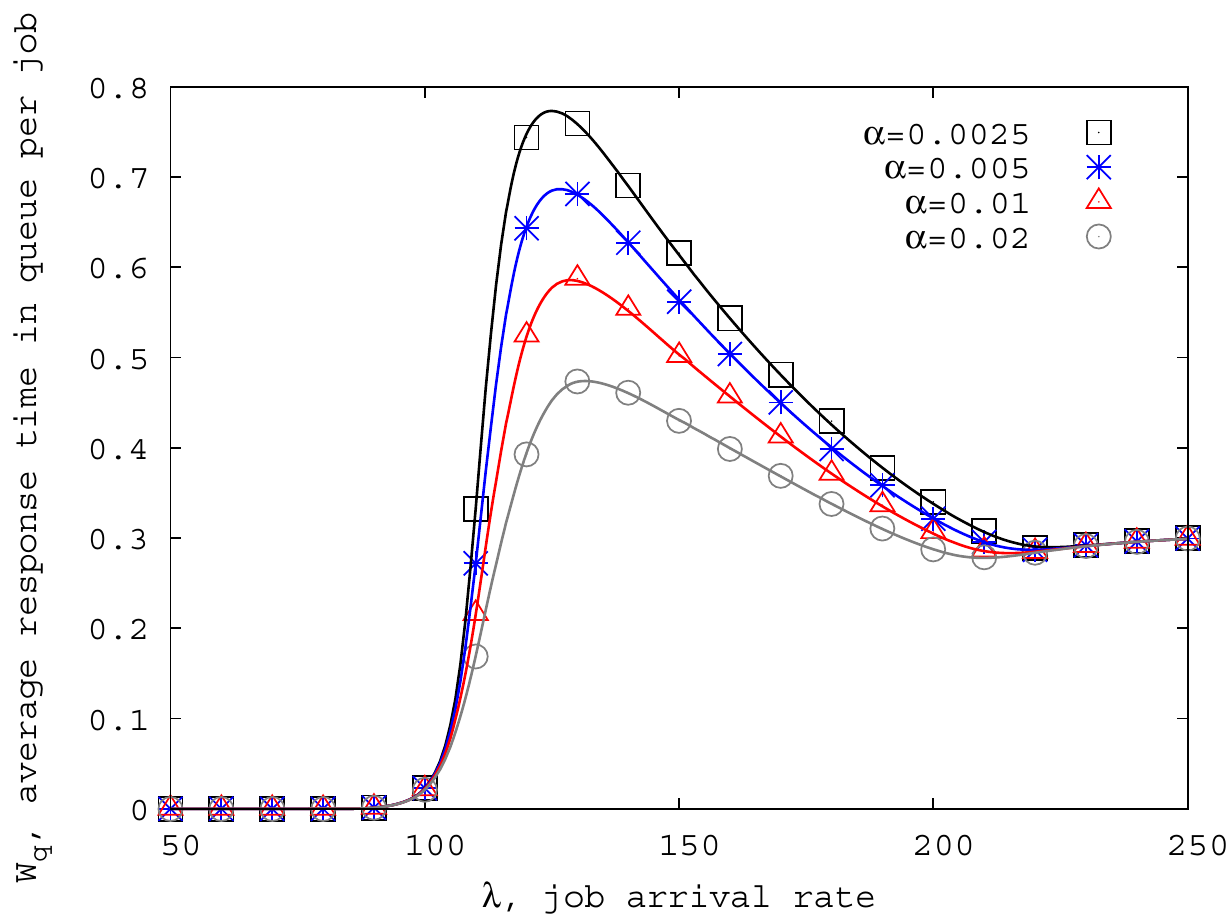}
		\caption{Impacts of $\alpha$ on $W_q$ ($k = 80$).}
		\label{fig:Wq_alpha}
\vspace{4mm}
	\end{subfigure}
	\begin{subfigure}[b]{0.48\textwidth}
		\includegraphics[width=8.5cm]{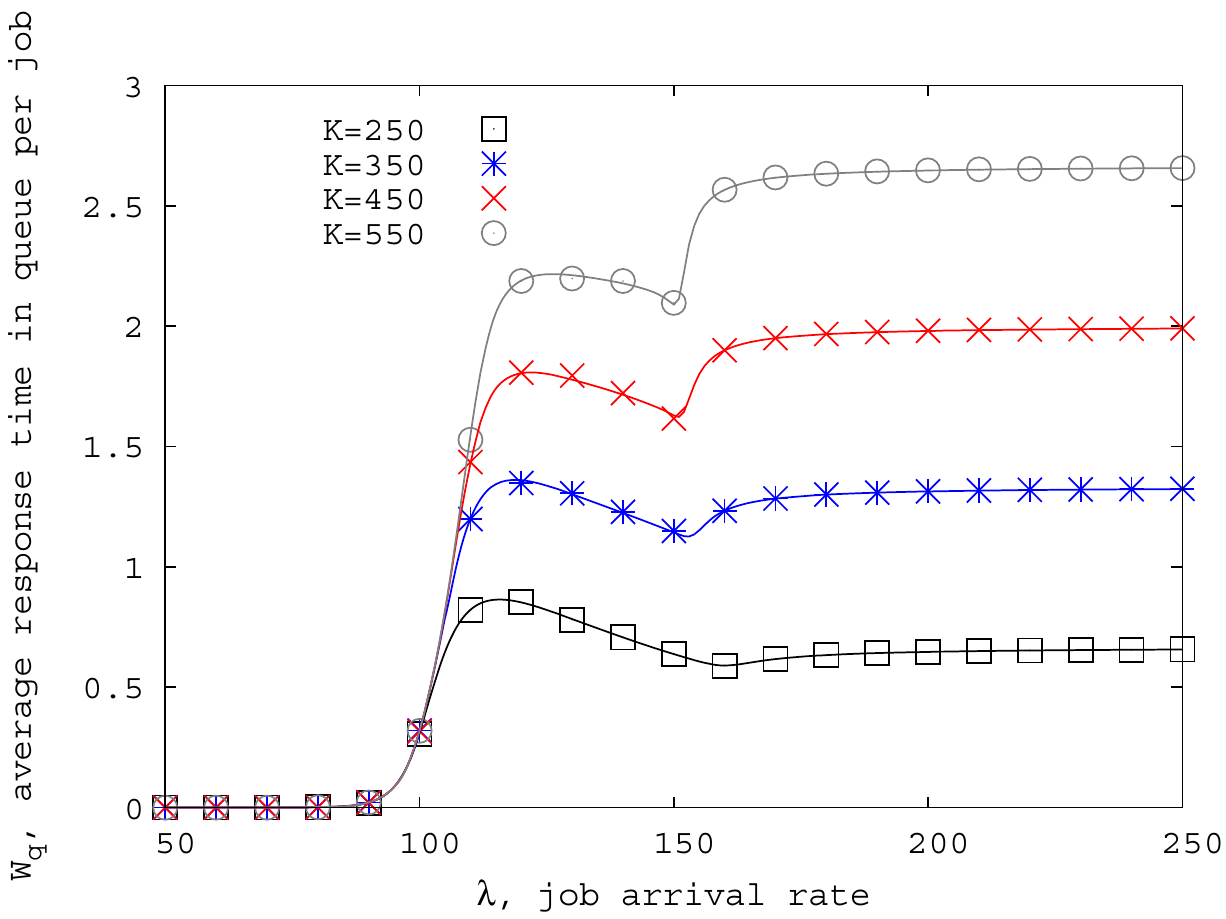}
		\caption{Impacts of $K$ on $W_q$ ($k = 50$).}
		\label{fig:Wq_K}
	\end{subfigure}
	\begin{subfigure}[b]{0.48\textwidth}
		\includegraphics[width=8.5cm]{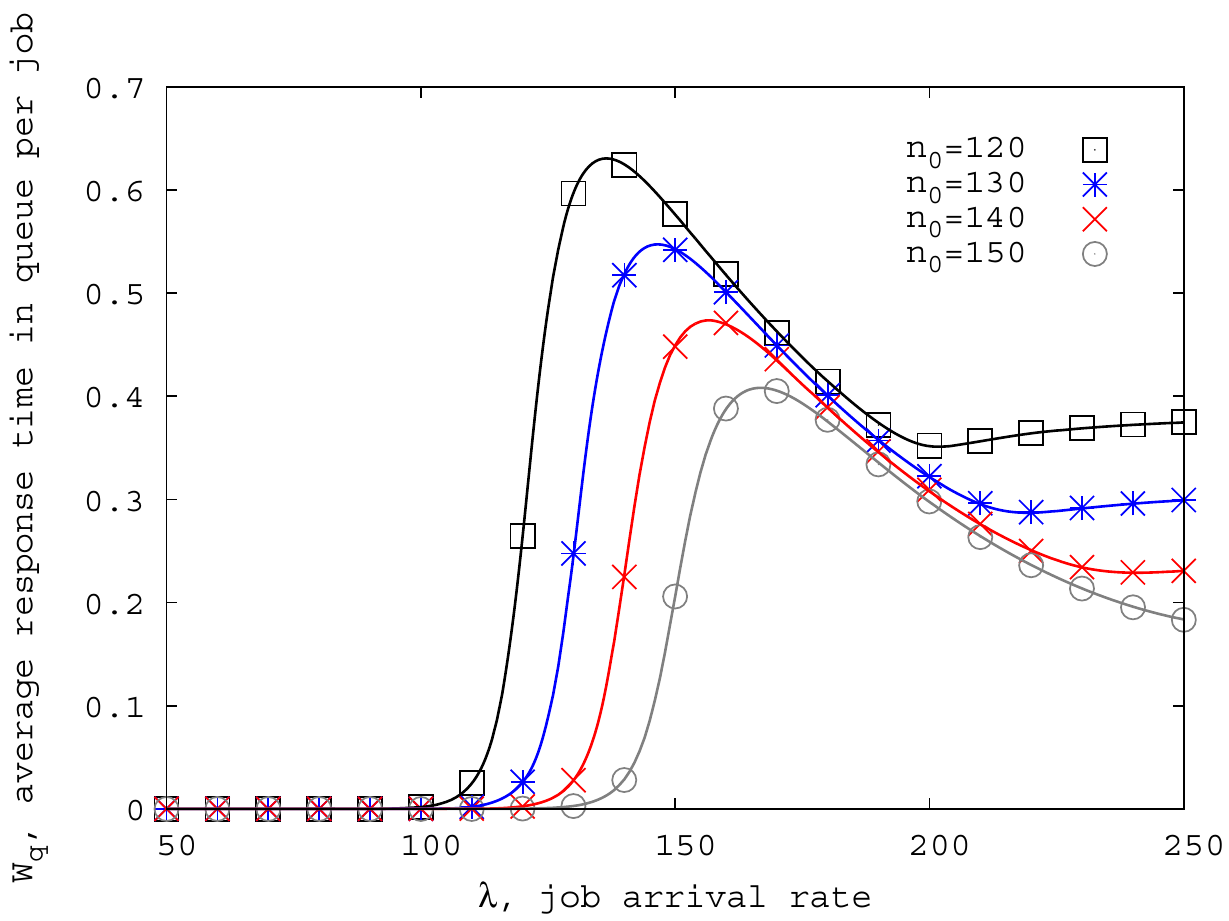}
		\caption{Impacts of $n_0$ on $W_q$ ($k = 60$).}
		\label{fig:Wq_n0}
	\end{subfigure}
	\caption{Impacts on $W_q$ while $1/\mu$, $1/\lambda$, and $1/\alpha$ are exponential distribution.}
	\label{fig:Impacts_Wq_expo}
\end{figure*}

\subsection{Impacts of arrival rate, $\lambda$}\label{ssec:lambda}
Figs.~\ref{fig:S_k}-\ref{fig:S_n0} show the impacts of $\lambda$ on $S$. Generally, one can see that $S$ is 0 at the beginning. It then grows sharply, and later raises smoothly and reaches at a bound when $\lambda$ increases. The reasons are as follows. When $\lambda \ll n_0\mu$, the incoming jobs are handled by the legacy equipment. No VNFs are turned on. Later,  VNFs are turned on as $\lambda$ approaches to $(n_0+k)\mu$. Accordingly, the server cost $S$ increases as $\lambda$ grows. When $\lambda>(n_0+k)\mu$, $S$ stops growing. Because all of the available $k$ VNFs are turned on, $S$ is bounded as $k$ VNFs' costs.

Figs.~\ref{fig:Wq_k}-\ref{fig:Wq_n0} illustrate the impacts of $\lambda$ on $W_q$. Interestingly, the trend of the curves can generally be divided into three phases\footnote{In Figs.~\ref{fig:Wq_alpha} and \ref{fig:Wq_n0}, only two phases are displayed due to the range of $\lambda$. Given a larger $\lambda$, all of the three phases will appear.}: ascent phase, descent phase, and saturation phase. In the first phase, $W_q$ grows sharply due to the setup time of VNFs. Specifically, when $\lambda \ll n_0\mu$, $W_q$ is almost 0 because all jobs are handled by legacy equipment. As $\lambda$ approaches to $n_0\mu$ and then is larger than $n_0\mu$, VNFs start to be turned on. In this phase, however, $W_q$ still raises due to the setup time of VNFs.  The reason is that VNFs just start to be turned on and do not reach their full capacities. In the second phase, we can see that $W_q$ starts to descend because the VNFs start serving jobs.
In the third phase, however, $W_q$ starts to ascend again and then saturate at a bound. The reason of ascent is that when $\lambda \geq (n_0+k)\mu$, the system is not able to handle the coming jobs. Finally, the curves go to saturation because the capacity of the system is too full to handle the jobs and the value of $W_q$ is limited by $K$.

\subsection{Impacts of the number of VNF instances, $k$}
Fig.~\ref{fig:S_k} shows the impacts of $k$ on $S$. All of the four curves increase and then reach a bound when $\lambda$ increase. A larger $k$ leads to a  bigger gap between the initial point and the bound.  When $\lambda$ increases, a larger $k$ means that more VNFs can be used to handle the growing job requests. Therefore, $S$ increases accordingly. If an operator wants to bound VNF budget to $S$, the operator can specify a suitable $k$ based on (\ref{eq:S}).

Fig.~\ref{fig:Wq_k} illustrates the impacts of $k$ on $W_q$. The impacts of $k$ are shown as the length of the second phase as discussed in Sec.~\ref{ssec:lambda}. The length of the second phase prolongs when $k$ increases. A larger $k$ gives the system more capability to handle the riseing job requests. That is, it delays the time that the system capacity reaches its bound. If an operator wants to bound the response time to $W_q$ of job request, the operator can choose a suitable $k$ based on (\ref{eq:W_q}).

\subsection{Impacts of VNF setup rate, $\alpha$}
Recall that $\alpha$ is the setup rate of VNFs. To change setup rate, one can adjust resources (e.g., CPU, memory) for VNFs. Fig.~\ref{fig:S_alpha} shows the impacts of $\alpha$ on $S$.  The impacts of $\alpha$ are shown as the slope of the curves. A larger $\alpha$ means smaller slope, but $\alpha$ has no effects at the beginning and the end of the curves. The reasons are as follows. A larger $\alpha$ means smaller VNF setup time. A smaller VNF setup time helps VNF to be turned on and to handle jobs faster so that the system is more efficient than that with larger VNF setup time.

Fig.~\ref{fig:Wq_alpha} illustrates the impacts of $\alpha$ on $W_q$. Again, the impacts of $\alpha$ are shown as the slope of curves. A larger $\alpha$ leads to smaller slopes. Also, $\alpha$ decides the maximum value of $W_q$. The reason is that smaller setup time enables VNF to handle jobs faster.

\subsection{Impacts of system capacity, $K$}
Fig.~\ref{fig:S_K} and Fig.~\ref{fig:Wq_K} depict the impacts of $K$ on $S$ and $W_q$, respectively. Based on our observation on Fig.~\ref{fig:S_K}, $K$ has limited impacts on $S$. As discussed in Sec.~\ref{ssec:lambda}, $S$ is mainly decided by $k$. As shown in Fig.~\ref{fig:Wq_K}, the impacts are significant on $W_q$. Different $K$ makes huge gaps between the curves. The curves also has three phases.  A larger $K$ leads to a larger $W_q$. The reason is that it enables more jobs waiting in the queue rather than dropping them.

\subsection{Impacts of legacy equipment capacity, $n_0$}
\label{ssec:n0}
Fig.~\ref{fig:S_n0} and Fig.~\ref{fig:Wq_n0} illustrate the impacts of $n_0$ on $S$ and $W_q$, respectively. We observe that the curves initiate at 0, and then fix at 0 for a period and start to grow when $\lambda$ increases. The value of $n_0$ decides the length of the period when the curves start to grow. The reason is that the legacy equipment can handle jobs within its capacity. When $\lambda$ exceeds the capacity of the legacy equipment, both $S$ and $W_q$ start to increase.

\subsection{Summary of Sections~\ref{ssec:lambda}--\ref{ssec:n0}}
Overall, Figs.~\ref{fig:Impacts_S_expo}-\ref{fig:Impacts_Wq_expo} not only demonstrate the correctness of our analytical model, but also show the impacts  of $\lambda$, $k$, $K$, $\mu$, $n_0$, $\alpha$ on the performance metrics $S$ and $W_q$.   Moreover, although service time $1/\mu$ is assumed to be exponential distributed, the proposed analytical model is also compatible for service time with deterministic, normal, uniform, Erlang, Gamma distribution. Due to page limit, more simulation results in terms of service time with various distributions are given in~\cite{YiRen2016Globecom}. Accordingly, the analytical model enables wide applicability in various scenarios. With our model, operators can quantify the performance easily. Therefore, it has important theoretical significance.

\subsection{Relation between $C$ and $k$}
In Figs.~\ref{fig:Impacts_S_expo}-\ref{fig:Impacts_Wq_expo}, we have demonstrated the correctness of our model. Now, we show how to obtain the optimal $k$ in various situations using the proposed model.

Figs.~\ref{fig:Impacts_op}(a) and (b) illustrate the results of two metrics in y-axis.  The left y-axis in  blue color is the cost function $C$ specified in (\ref{eq:optimal}), which is corresponding to the blue curved line in the figure. The right y-axis in red color is the average response time in the queue, $W_q$, which is corresponding to the dotted curved line in red color. The x-axis is the number of VNF instances, $k$.

Let us first take a look at the red dotted line in Fig.~\ref{fig:Impacts_op}(a). It is shown that $W_q$ decreases when $k$ grows. Increasing $k$ will make arriving jobs experience lower waiting time and better QoS, and decrease SLA violation. Specifically, we observe that $W_q$ first declines sharply, and then at a point (around $k=28$) it starts to decrease slowly. At the same point, however, the cost function $C$ shown in blue curved line starts to grow sharply, leading to higher cost. The optimal $k$ is 28 when $\lambda=130$ (job/s). By applying the optimal $k$, the corresponding $W_q$ is 1.17 s (see the green dotted straight line). If a mobile operator defines a latency greater than 1.17~s as SLA violation, that is, $W_q'$ in (\ref{eq:optimal}) is set to be greater than 1.17~s, we can just set $k$ as 28 to perfectly balance $C$ and fulfill the latency requirement.  Otherwise, we can use the green dotted line to find the $k$ which is corresponding to the latency less than the requirement set by the operator. In this case, although the $C$ is not the minimum, the $k$ is the optimal value that can satisfy the latency requirement. Similarly, one can find another example of setting the optimal $k$ with $\lambda=170$ in Fig.~\ref{fig:Impacts_op}(b).

Here, we only demonstrate that any latency requirement can be met by setting the optimal value $k$ by using (\ref{eq:optimal}). Please note that other parameters, such as $n_o$, $\mu$, $K$, and $\alpha$ in $\tau$ can be easily applied to (\ref{eq:optimal}) or (\ref{eq:optimal_xxx})  using the same method. Due to space limit, we do not explain it in this paper.

\begin{figure*}[t]
	\centering
	\begin{subfigure}[b]{0.48\textwidth}
		\includegraphics[width=8.5cm]{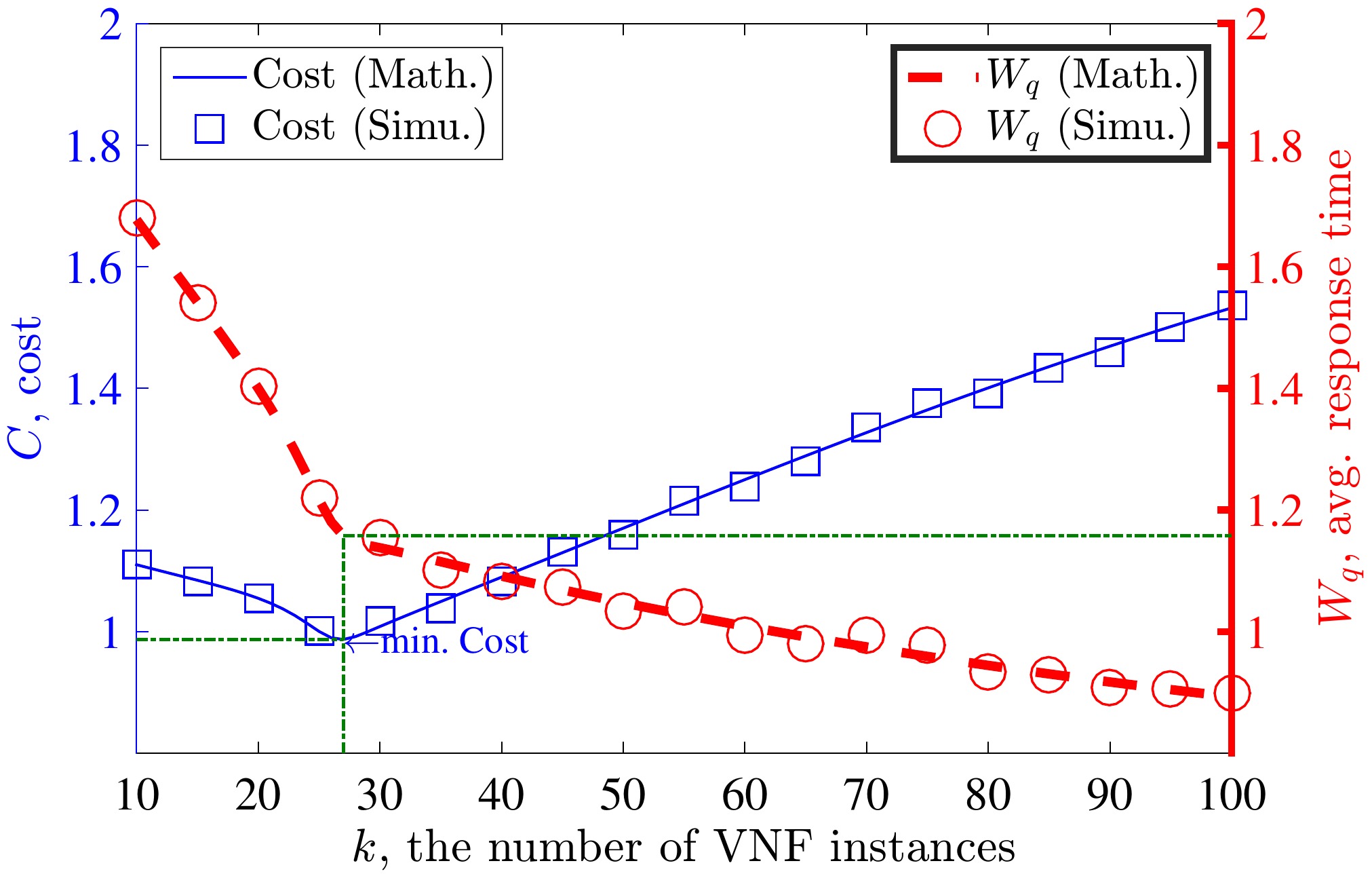}
		\caption{$\lambda=130$}
		\label{fig:Impact_lambda}
	\end{subfigure}
	\begin{subfigure}[b]{0.48\textwidth}
		\includegraphics[width=8.5cm]{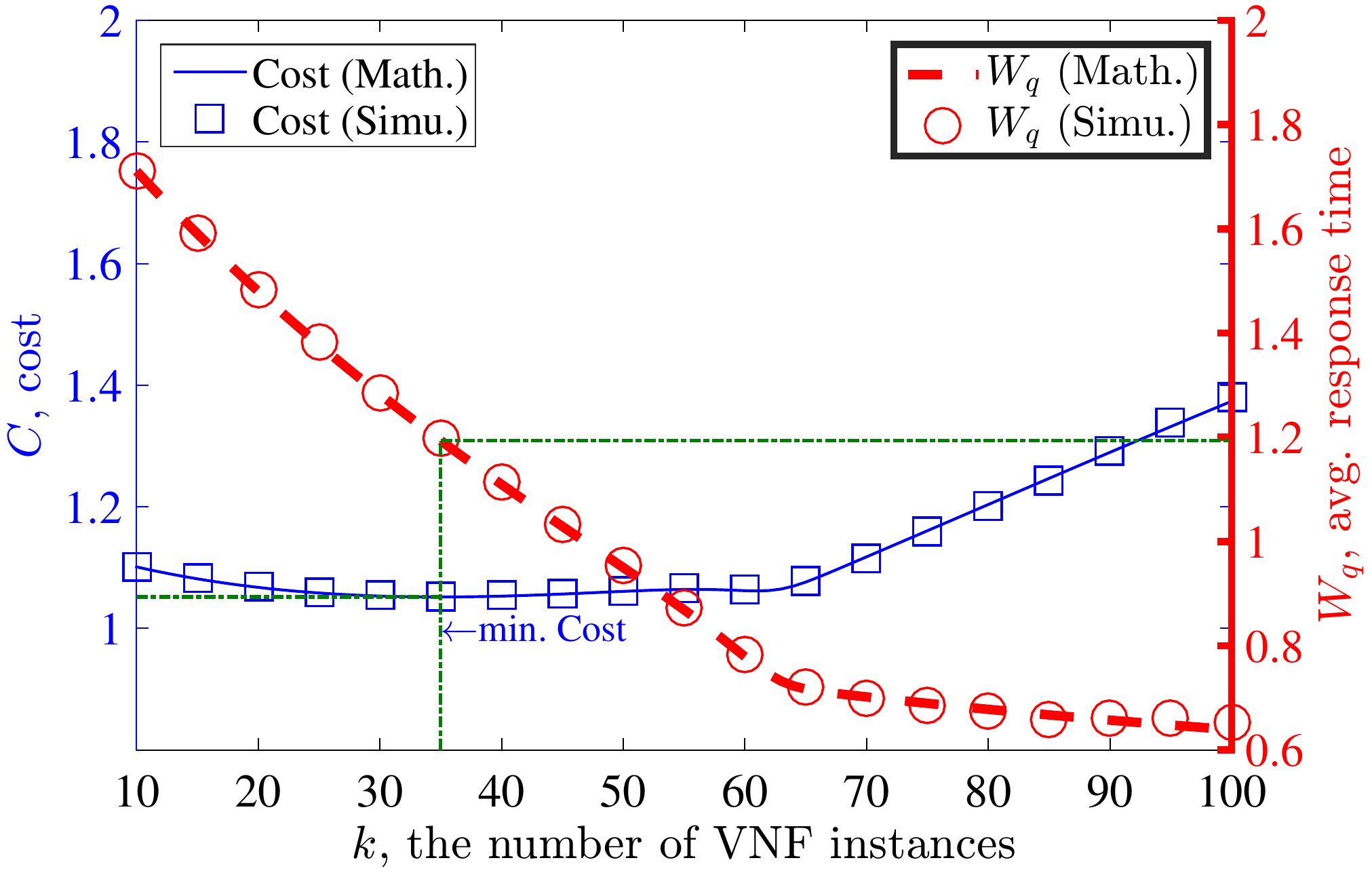}
		\caption{$\lambda=170$}
		\label{fig:Impact_K}
	\end{subfigure}
	\caption{Optimal $k$ in various conditions.}
	\label{fig:Impacts_op}
\end{figure*}

%
%
%
%
%

\section{Conclusions} \label{sec:Conclusions}
In this paper, we propose DASA to address the tradeoff between performance and operation cost. We develop analytical and simulation models to study the average job response time $W_q$ and operation cost $S$. The results quantify the performance metrics. Our model fills the research gap by taking both  VNF setup time and  the capacity of legacy equipment into consideration in vEPC. Besides, we propose a novel recursive algorithm to reduce the computational complexity significantly. Our study provides guidelines for mobile operators to bound their operation cost and configure job response time in a systematic way. Based on our performance study, operators can design optimization strategies to quickly obtain the operation cost and system performance without real deployment to save cost and time.

For future work, one extension is to generalize the VNF setup time, and the arrival time and service time. Right now, there is no literature to support that when they are exponential random variables. These results could be generalized by Markovian Arrival Processes~\cite{pender2016approximations} or approximated by using orthogonal polynomial approaches~\cite{pender2014gram}. Also, we plan to relax the assumption of VNF scaling in/out capability, that is, from one VNF instance per time to arbitrary numbers of instances per time. Moreover, another extension is to consider impatient customer. That is, waiting jobs have time-limit. After that, they will be discarded from the system.

\section*{Acknowledgments}
We especially thank Zheng-Wei Yu and Yi-Hao Lin for their help in our simulations.

\bibliographystyle{IEEEtran}
\bibliography{IEEEACM,NFV}

\appendices

%
%

\ifCLASSOPTIONcaptionsoff
  \newpage
\fi



%

%

\vspace{10mm}
\begin{IEEEbiography}[{\includegraphics[width=1in,height=1.25in,clip,keepaspectratio]{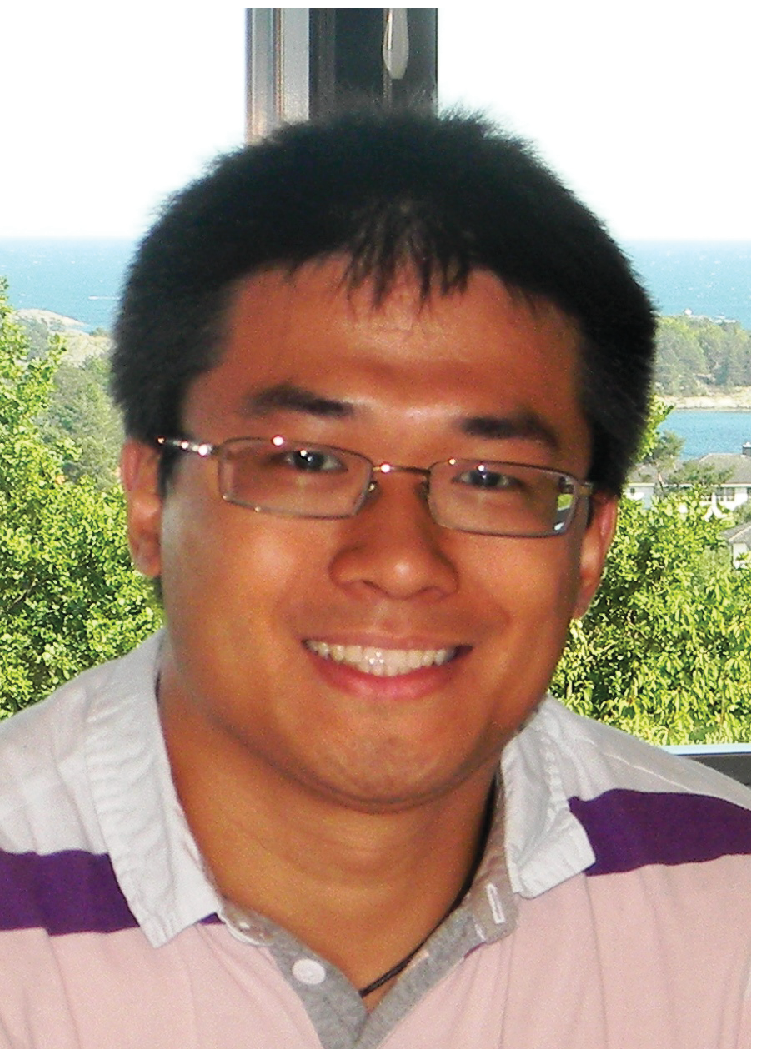}}]{Yi Ren} (S'08-M'13)
	has been an Assistant Researcher at National Chiao Tung University (NCTU), Taiwan since 2012. He received his Ph.D. in Information Communication and Technology from the University of Agder (UiA), Norway in April 2012. His current research interests include security and performance analysis in wireless sensor networks, ad hoc, and mesh networks, LTE, smart grid, and e-health security. He received the Best Paper Award in IEEE MDM 2012.
	\vspace{-10mm}
\end{IEEEbiography}

\vspace{10mm}
\begin{IEEEbiography}[{\includegraphics[width=1in,height=1.25in,clip,keepaspectratio]{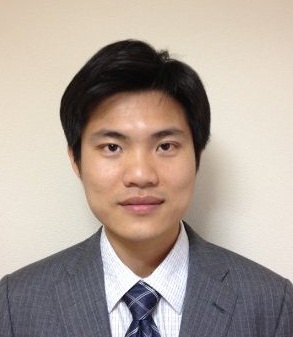}}]{Tuan Phung-Duc}  is an Assistant Professor at Faculty of Engineering, Information and Systems, University of Tsukuba. He received a Ph.D. in Informatics from Kyoto University in 2011. He is in the Editorial Board of the KSII Transactions on Internet and Information Systems and two other  international journals. He served a Guest Editor of the special issue of Annals of Operations Research on Retrial Queues and Related Models and currently is serving as a Guest Editor of the Special Issue of the same journal on Queueing Theory and Network Applications. He was the Chairman of 10th International Workshop on Retrial Queues (WRQ'2014) and the TPC co-chair of 23rd International Conference on Analytical and Stochastic Modelling Techniques and Applications (ASMTA'16). Dr. Phung-Duc received the Research Encourage Award from The Operations Research Society of Japan in 2013. His research interests include Applied Probability, Stochastic Models and their Applications in Performance Analysis of Telecommunication and Service Systems.
	\vspace{-10mm}
\end{IEEEbiography}

\begin{IEEEbiography}[{\includegraphics[width=1in,height=1.25in,clip,keepaspectratio]{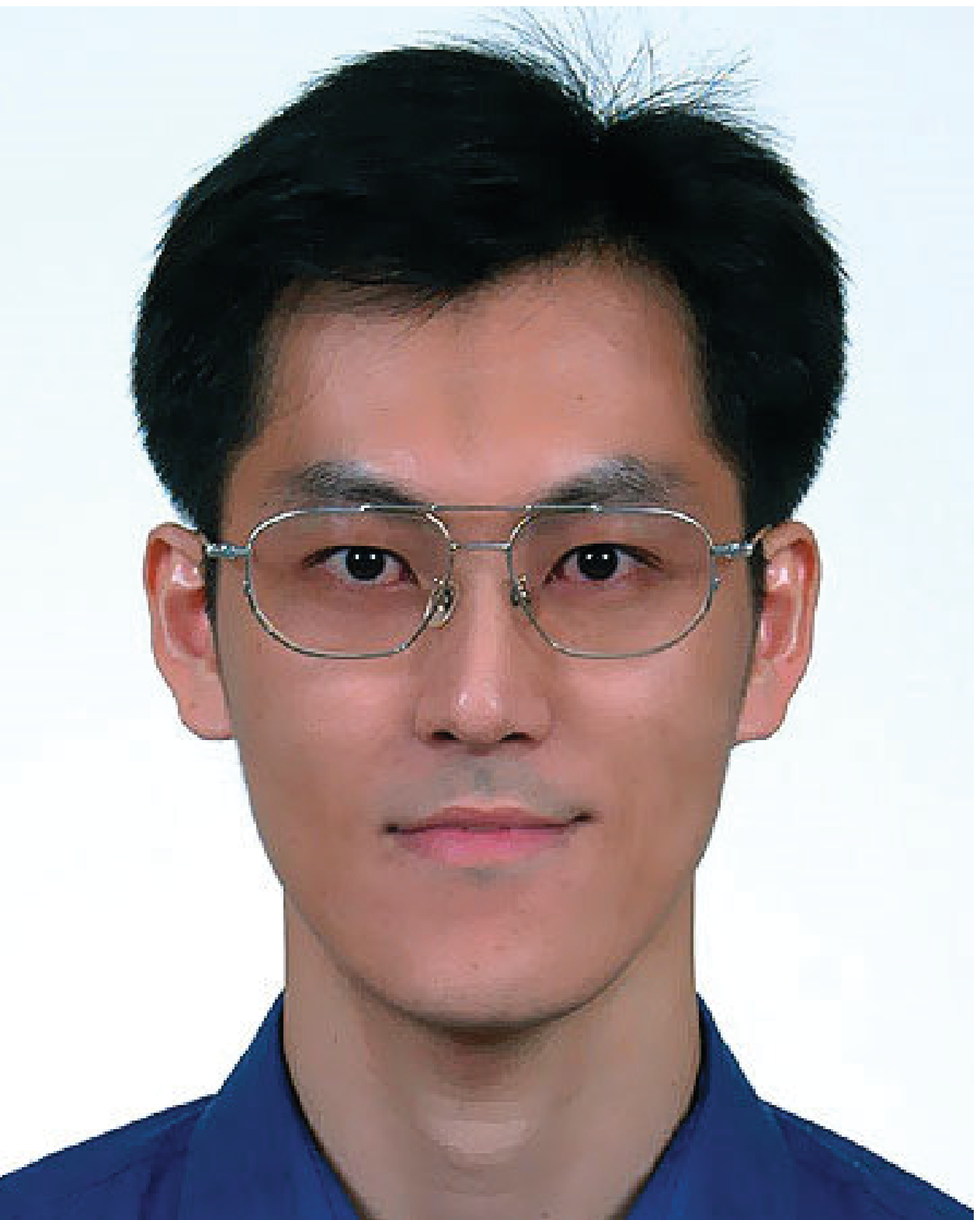}}]{Jyh-Cheng
		Chen} (S'96-M'99-SM'04-F'12) received the Ph.D. degree from the State
	University of New York at Buffalo, USA, in 1998.
	
	He was a Research Scientist with Bellcore/Telcordia Technologies,
	Morristown, NJ, USA, from 1998 to 2001, and a Senior Scientist with
	Telcordia Technologies, Piscataway, NJ, USA, from 2008 to 2010. He was
	with the Department of Computer Science, National Tsing Hua University
	(NTHU), Hsinchu, Taiwan, as an Assistant Professor, an Associate
	Professor, and a Professor from 2001 to 2008. He was also the Director
	of the Institute of Network Engineering with National Chiao Tung
	University (NCTU), Hsinchu, from 2011 to 2014. He has been a Faculty
	Member with NCTU since 2010. He is currently a Distinguished Professor
	with the Department of Computer Science, NCTU. He is also serving as the
	Convener, Computer Science Program, Ministry of Science and Technology,
	Taiwan.
	
	Dr. Chen received numerous awards, including the Excellent Teaching
	Award from NCTU, the Outstanding I. T. Elite Award, Taiwan, the Mentor
	of Merit Award from NCTU, the K. T. Li Breakthrough Award from the
	Institute of Information and Computing Machinery, the Outstanding
	Professor of Electrical Engineering from the Chinese Institute of
	Electrical Engineering, the Outstanding Research Award from the Ministry
	of Science and Technology, the Outstanding Teaching Award from NTHU, the
	best paper award for Young Scholars from the IEEE Communications Society
	Taipei and Tainan Chapters, and the IEEE Information Theory Society
	Taipei Chapter, and the Telcordia CEO Award. He is a Fellow of the IEEE
	and a Distinguished Member of the ACM. He is a member of the Fellows
	Evaluation Committee, IEEE Computer Society, 2012 and 2016.
\end{IEEEbiography}\vfill



\section*{Supplementary material}

# Appendix

## A Variance of the Cost Function

Please note that one could envisage different operational interests. For example, if an operator also cares about the job blocking probability $P_b$, the mean waiting time in the system $W$, and the mean number of jobs in the system $L$, the cost function $C$ in (1) is replaced by:

$$C = w_1 W_q + w_2 S + w_3 P_b + w_4 W + w_5 L \tag{40}$$

where $w_3$, $w_4$, and $w_5$ are weighting factors for $P_b$, $W$, and $L$, respectively. Similarly, since the closed forms of these metrics are derived in Equations (36), (38), (37), (35), and (34), one can easily find the local minimum when $C' = 1$ and $C'' > 0$ are satisfied.

With (40), (39) can be written as:

$$\begin{aligned} \underset{\tau}{\arg\min} \quad & C = w_1 W_q + w_2 S + w_3 P_b + w_4 W + w_5 L \ , \\ \text{subject to} \quad & 0 < W_q < W_q' \ . \end{aligned} \tag{41}$$

where $\tau \in \{k, n_0, \mu, K, \alpha\}$.

The model can be easily extended to any number of parameters.

## Impacts of the Various Distribution of Service Time

Here, extensive simulation results in terms of service time $1/\mu$ with various distributions are introduced. Since there is no empirical evidence for the distribution of $1/\mu$, in the proposed analytical model it is assumed to be exponential distribution for the sake of simplicity. According to the results, the proposed analytical model is compatible for service time with deterministic, normal, Erlang, Gamma, Uniform distribution.

The proposed analytical model is compatible with service time with deterministic distribution, as shown in Figs. 10 and 11.

The proposed analytical model is compatible with service time with normal distribution, as shown in Figs. 12 and 13.

The above results demonstrate that the proposed model can be used for modeling the system with the aforementioned service time distribution.

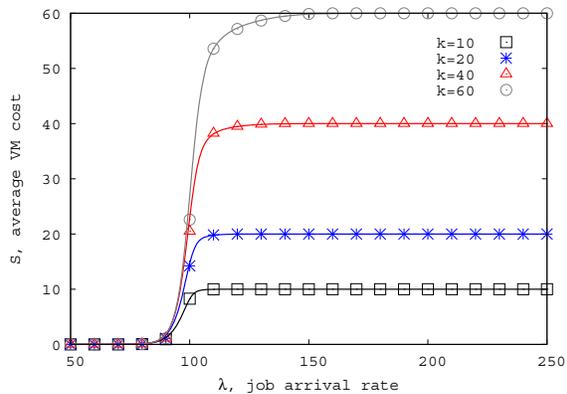

(a) Impacts of $k$ on $S$ ($n_0 = 100$).

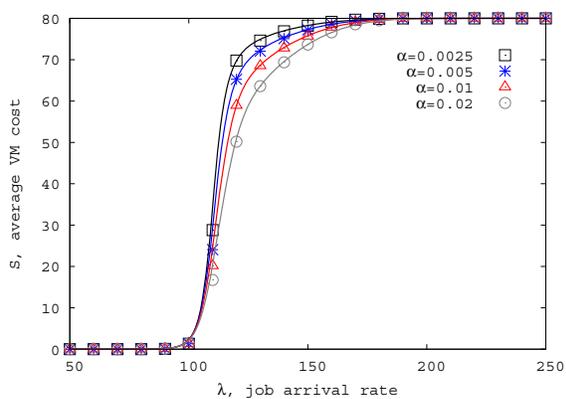

(b) Impacts of $\alpha$ on $S$ ($k = 80$).

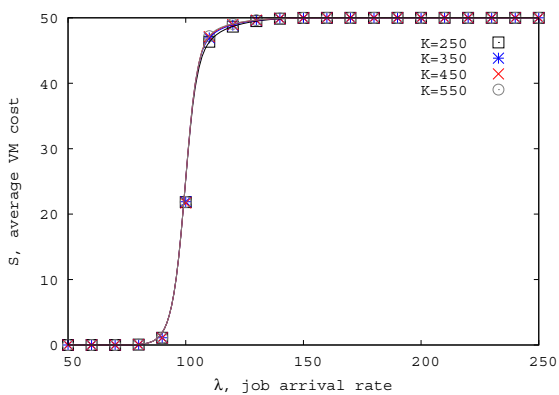

(c) Impacts of $K$ on $S$ ($k = 50$).

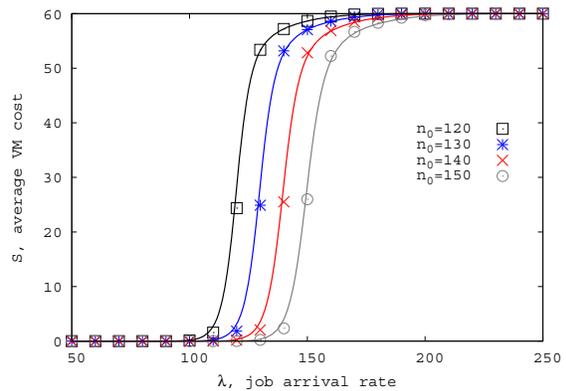

(d) Impacts of $n_0$ on $S$ ($k = 60$).

Figure 10: Impacts on $S$ while service time is fixed.

2

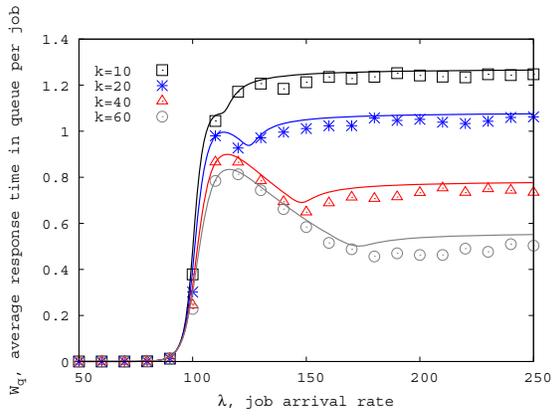

(a) Impacts of $k$ on $Wq$ ($n_0 = 100$).

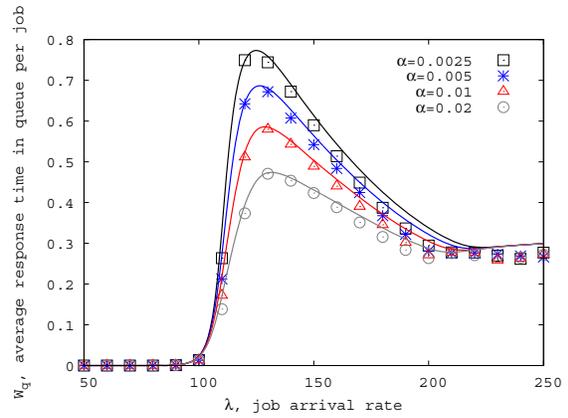

(b) Impacts of $\alpha$ on $Wq$ ($k = 80$).

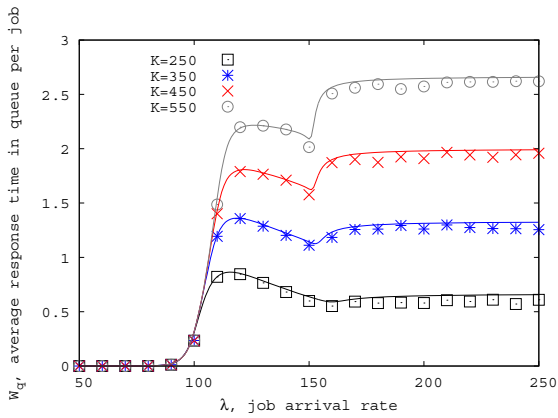

(c) Impacts of $K$ on $Wq$ ($k = 50$).

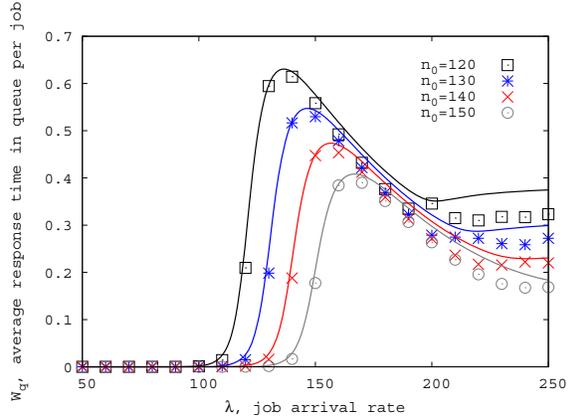

(d) Impacts of $n_0$ on $Wq$ ($k = 60$).

Figure 11: Impacts on $Wq$ while service time is fixed.

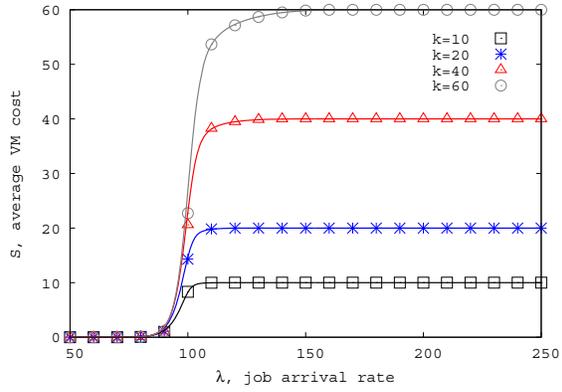

(a) Impacts of $k$ on $S$ ($n_0 = 100$).

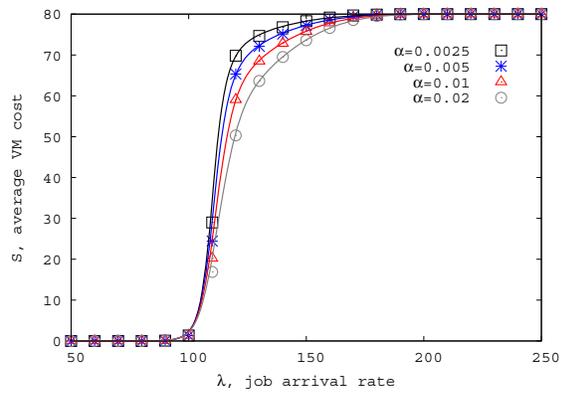

(b) Impacts of $\alpha$ on $S$ ($k = 80$).

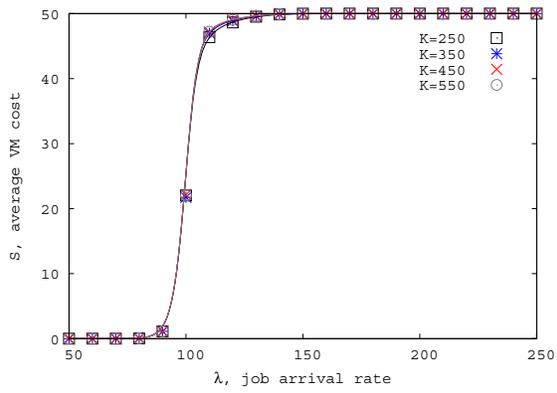

(c) Impacts of $K$ on $S$ ($k = 50$).

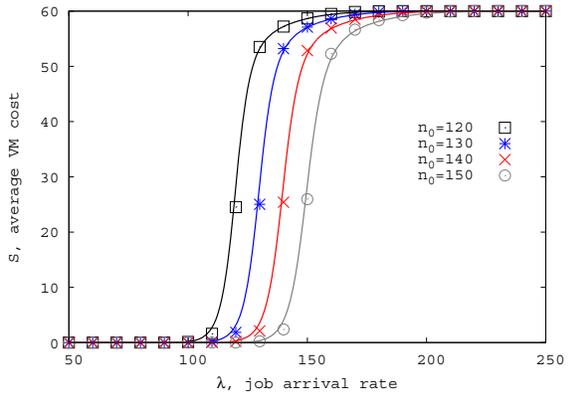

(d) Impacts of $n_0$ on $S$ ($k = 60$).

Figure 12: Impacts on $S$ while service time is normal distribution with $\sigma = 0.1$.

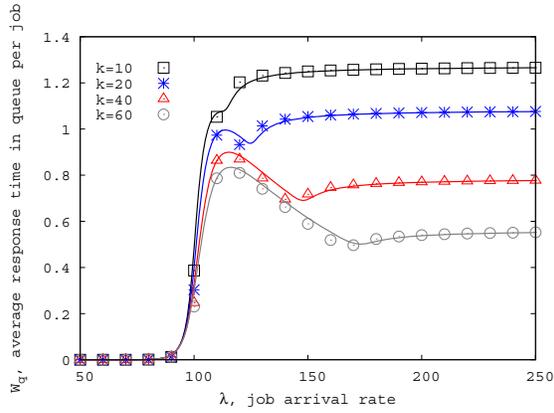

(a) Impacts of $k$ on $Wq$ ($n_0 = 100$).

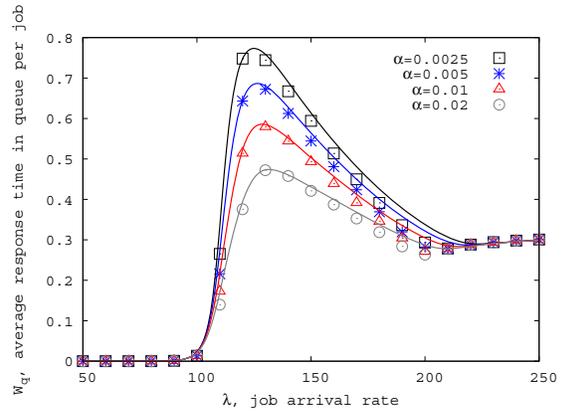

(b) Impacts of $\alpha$ on $Wq$ ($k = 80$).

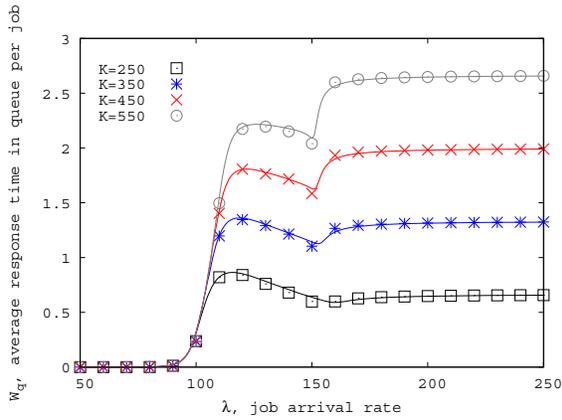

(c) Impacts of $K$ on $Wq$ ($k = 50$).

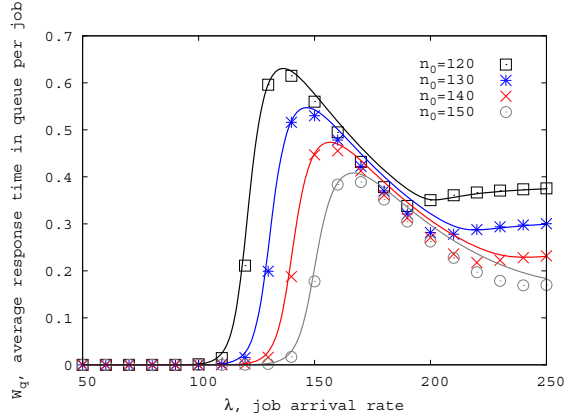

(d) Impacts of $n_0$ on $Wq$ ($k = 60$).

Figure 13: Impacts on $Wq$ while service time is normal distribution with $\sigma = 0.1$.



\end{document}